\documentclass[preprint, 3p]{elsarticle}

\usepackage[english]{babel}
\usepackage[utf8]{inputenc}
\usepackage{amsthm}
\usepackage{amsmath}
\usepackage{amssymb}
\usepackage{hyperref}
\usepackage{enumerate}
\usepackage{xcolor}
\usepackage{color,soul}
\usepackage{tikz}
\usepackage{placeins}
\usepackage{subfigure} 
\usepackage{booktabs}
\usepackage[ruled,vlined]{algorithm2e}
\usepackage[titletoc]{appendix}

\journal{TBA}

\newtheorem{theorem}{Theorem}[section]
\newtheorem{corollary}[theorem]{Corollary}
\newtheorem{lemma}[theorem]{Lemma}

\newtheorem{definition}[theorem]{Definition}
\newtheorem{remark}[theorem]{Remark}
\numberwithin{equation}{section}

\DeclareMathOperator{\sign}{sign}
\DeclareMathOperator{\tr}{tr}

\DeclareMathOperator{\argmax}{argmax}

\makeatletter
\DeclareRobustCommand\widecheck[1]{{\mathpalette\@widecheck{#1}}}
\def\@widecheck#1#2{%
    \setbox\z@\hbox{\m@th$#1#2$}%
    \setbox\tw@\hbox{\m@th$#1%
       \widehat{%
          \vrule\@width\z@\@height\ht\z@
          \vrule\@height\z@\@width\wd\z@}$}%
    \dp\tw@-\ht\z@
    \@tempdima\ht\z@ \advance\@tempdima2\ht\tw@ \divide\@tempdima\thr@@
    \setbox\tw@\hbox{%
       \raise\@tempdima\hbox{\scalebox{1}[-1]{\lower\@tempdima\box
\tw@}}}%
    {\ooalign{\box\tw@ \cr \box\z@}}}
\makeatother

\begin{document}
	
\begin{frontmatter}
	
\title{On regularized optimal execution problems and their singular limits}

\author[IMEUFF]{Max O. Souza}
\ead{maxsouza@id.uff.br}

\author[IMEUFF]{Y. Thamsten}
\ead{ythamsten@id.uff.br}

\address[IMEUFF]{Instituto de Matem\'atica e Estat\'istica, Universidade Federal Fluminense, Niter\'{o}i, RJ, Brasil}

\date{}
	
	
\begin{abstract}
    We investigate the portfolio execution problem under a framework in which volatility and liquidity are both uncertain. In our model, we assume that a multidimensional Markovian stochastic factor drives both of them. Moreover, we model indirect liquidity costs as temporary price impact, stipulating a power law to relate it to the agent's turnover rate. We first analyze the regularized setting, in which the admissible strategies do not ensure complete execution of the initial inventory. We prove the existence and uniqueness of a continuous and bounded viscosity solution of the Hamilton-Jacobi-Bellman equation, whence we obtain a characterization of the optimal trading rate. As a byproduct of our proof, we obtain a numerical algorithm. Then, we analyze the constrained problem, in which admissible strategies must guarantee complete execution to the trader. We solve it through a monotonicity argument, obtaining the optimal strategy as a singular limit of the regularized counterparts.
\end{abstract}	


\end{frontmatter}	

\section{Introduction} \label{sec:Intro}

We refer to the situation in which a financial agent must execute a large trade as the execution problem. It constitutes a significant part of algorithmic trading. A paradigm that we can use to find a solution is to set up a trading schedule: stipulate a time horizon to complete the program, and fractionate the originally large order into smaller ones. In this setting, to identify an adequate rate of trading, we must take into account the trade-off among two major financial complexities: transaction costs; exposure to price risk. We describe how these effects influence trading decisions as follows. On the one hand, if we send too sizeable orders, then we will undergo considerable price impact.\footnote{As put in \cite{bouchaud2009price}, ``\textit{price impact} refers to the correlation between an incoming order (to buy or to sell) and the subsequent price change''.} For instance, in this case, our behavior may demand the consumption of many layers of liquidity of the Limit Order Book (LOB) at each time we execute a trade. On the other hand, if we take too long to complete execution, then we will be excessively exposed to the inherent uncertainty in the price movements.


Systematic research on optimal execution began with the works of Almgren and Chriss (AC) \cite{almgren1999value,almgren2001optimal} — see also \cite{bertsimas1998optimal}. The AC model has as among its premises the following core assumptions: (i) Bachelier price dynamics; (ii) price impacts are of two types, viz., a temporary and a permanent one, depending on their persistence in time; (iii) the trader only sends market orders. This model proved itself to be very useful as a framework for the investigation of the trade execution problem. Some efforts reevaluating (ii) include \cite{almgren2003optimal,curato2017optimal,gatheral2012transient}. Regarding (iii), papers \cite{cartea2015optimal,gueant2012optimal} consider the use of limit orders in execution algorithms. A geometric Brownian motion models the stock price in \cite{gatheral2011optimal}, in contradistinction to assumption (i). For extensions envisaging to account for a non-constant market volume, see \cite{cartea2016incorporating,gueant2016financial}.

In the standard continuous-time model of AC, the temporary and permanent price impacts are functions of the speed of trading of the agent. In this direction, the monographs \cite{cartea2015algorithmic,sophie2018market} contain statistical analyses of micro-structural aspects that are relevant to the study of algorithmic trading. We remark that they corroborate, to an extent, the adequacy of the AC setting with reality. A popular choice for the functional form of the impacts is the linear model: one assumes that the impacts per share are proportional to the size of the orders sent to the market by the trader (usually with different slopes). Under suitable hypotheses, Huberman and Stanzl proved in \cite{huberman2004price} that the absence of quasi-arbitrage implies that the permanent price impact must indeed manifest in this way. This assumption also rules out price manipulation, as well as dynamic arbitrage, see \cite{gatheral2013dynamical,gueant2016financial}. 

There are many insightful advances in the area of market microstructure research regarding the study price impact; we refer to \cite{bouchaud2009price,sophie2018market} and the references therein for a broad account of these. For the sake of analytical tractability, some works tend to consider a linear model for temporary and permanent price impacts, leading to a mean-variance framework, see for instance the seminal paper \cite{almgren1999value} and also \cite{cartea2016incorporating,bank2017hedging,ekren2022utility}. In general, however, assuming a nonlinear price impact instead of the linear model --- such as a power law model --- leads to a more realistic shape of the impact curves, see \cite{almgren2005direct, bacry2015market}. In particular, in settings in which explicit solutions are not available in the linear case, such as the one we will investigate in this work, it is pertinent to consider this more general assumption, resorting to numerical methods to compute the solutions. The references \cite{cartea2015algorithmic,gueant2016financial} contain a few works dealing with nonlinear temporary price impact. See also \cite{almgren2003optimal} for an early development of the modeling in this direction, and \cite{leal2020learning} for an application of neural networks to the execution problem with a power law assumption on the nonlinearity of the cost under discussion.


In the AC setting, we find the optimal execution strategy through an optimization problem in a mean-variance framework. The design of the performance criteria is done in a way to represent a measure of execution quality. More precisely, it consists of minimizing the mean and the variance of the difference between the realized revenue and a benchmark price. A common choice for the latter is the pre-trade price, leading to a criterion known as implementation shortfall (IS), which we will consider henceforth. Alternatives are volume-weighted average price (VWAP), target close (TC), and percentage of volume (POV); we refer to \cite{cartea2015algorithmic,gueant2016financial} for in-depth discussions about these criteria.

It is a stylized fact that a flat or known deterministic profile for the volatility and liquidity is a suitable assumption for highly liquid assets, such as large-capitalization US stocks. It leads, in turn, to deterministic closed-form formulas for the corresponding optimal strategies. An advantage of this fact is that we can pre-compute the trading curve. However, as it is described in \cite{almgren2012optimal}, small-cap stocks are more difficult to trade. Throughout the trading day, there are moments in which trading is cheap, and others in which it is expensive. Furthermore, these moments alternate randomly, making the problem more difficult. Thus, when performing execution of portfolios containing less liquid assets, considering uncertainty in their volatility and liquidity becomes an important feature. In effect, an illiquid scenario ought to lead to higher price impacts; this, together with lower volatility levels, incentivizes the trader to slow down the execution program. Analogously, a situation of higher liquidity and volatility urges her to speed up. In these circumstances, the need for strategies that adapt to the current market state arises. We stress that there are also many studies on intra-day volatility estimation, of which we mention \cite{gatheral2010zero}. Thus, implementing these adaptive strategies in practice seems to be a feasible endeavor.

In the present work, we consider a risk averse agent who assesses the performance of the admissible strategies by utilizing criteria of the IS kind. However, two distinguishing modelling features here are that we allow the temporary impact to be of a general power law type, and that, at a first moment, we tolerate hitting the terminal target only approximately. Then, we consider the constrained problem, i.e., the framework in which the agent must necessarily finish with zero inventory. Moreover, we consider volatility and the temporary price impact parameter as stochastic processes, with a multidimensional Markov diffusion as their driver. In the literature, a paper considering stochastic price impacts, but constant volatility, is \cite{barger2019optimal}. Also, this work restrains to linear impacts. There is a model for treating the problem of slicing a VWAP order under stochastic volatility and market volume, but with no form of execution costs, in \cite{konishi2002optimal}. The framework of \cite{graewe2017optimal} considers constant temporary and permanent price impacts, but stochastic resilience and volatility (the latter through the uncertain risk aversion parameter). Some advances in the execution problem when simultaneously considering stochastic liquidity and volatility are \cite{almgren2012optimal,graewe2015non,graewe2018smooth,horst2020continuous,horst2019portfolio}, for the single agent case, \cite{cheridito2014optimal}, in a discrete time setting, and \cite{evangelista2020finite,fu2020mean}, in game-theoretic models. 

We prove the existence and uniqueness of viscosity solutions of the Hamilton-Jacobi-Bellman (HJB) Partial Differential Equation (PDE) for the penalized problem, in an appropriate class, via an iterative monotonicity fixed-point method. We obtain appropriate bounds for the solution of the main HJB, and this allows us to derive bounds for the terminal inventory in terms of the penalization parameter. In particular, we can identify a range of values of the penalization parameter which guarantee that following the corresponding optimal trading rate leads to the execution of any prescribed portion of the trader's initial inventory. We also show that this strategy does not lead to speculative trading. Moreover, our method naturally yields a numerical method for computing the solution, which we implement to provide illustrations. Then, we use another monotonicity argument to establish that the singular limit, as the terminal penalization parameter goes to infinity, is the solution of the constrained problem, i.e., the one in which we require strict execution. 

The works that are closer to the present one in their modeling aspects are \cite{almgren2012optimal,graewe2018smooth,horst2020continuous}. In \cite{almgren2012optimal,horst2020continuous}, the authors only investigate the case of a linear temporary price impact per share, leading to quadratic objective functionals. The work \cite{graewe2018smooth} is based on criteria stemming from the power law hypothesis, as in the present work. Our results are in the line of those of \cite{graewe2018smooth,horst2020continuous}, but our proofs rely on techniques that are distinct to the ones employed there. In contradistinction to those papers, we analyze the singular problem as the limit of the regularized counterpart. Our approach has the benefit that we only require continuity of the coefficients to prove our existence and uniqueness result. Also, the method we employ yields a numerical algorithm, cf. \cite[Remark 2.9]{horst2020continuous}.

We organize the remainder of this paper as follows. In Section \ref{sec:TheModel} we describe the model in detail, and define the value functions for both problems, viz., the regularized and the constrained ones. In Section \ref{sec:AnalysisPDEs}, we make an ansatz leading to a semilinear HJB PDE for the corresponding value function. We observe that Section \ref{sec:AnalysisPDEs} is mainly technical, so that one who is not interested in the proofs can skip directly to its last subsection, where results are numerically assessed. We prove existence and uniqueness of a viscosity solution of it, in a suitable class. We also prove some properties of the optimal turnover rate, and present some numerical experiments. In Section \ref{sec:singular} we obtain the solution of the constrained problem as a singular limit of the penalized optimal strategies via a monotonicity argument. In Section \ref{sec:Conclusions} we present our conclusions.

We finish this introductory section by fixing some notations. Throughout the present work, we fix $T>0$ to represent our finite trading horizon. We denote by $\left(\Omega,\mathcal{F},\mathbb{F},\mathbb{P}\right)$ a complete filtered probability space, where $\mathbb{F} = \left\{ \mathcal{F}_t \right\}_{0\leqslant t \leqslant T}$ is a filtration satisfying the usual conditions, and $\mathcal{F}_T = \mathcal{F}.$ We also assume that $\mathbb{F}$ supports a one-dimensional Brownian motion $B,$ as well as an $m-$dimensional one $\textbf{W},$ for a fixed positive integer $m.$ All stochastic processes figuring throughout this work will be $\mathbb{F}-$adapted. We interpret $\mathbb{P}$ as the historical (or statistical) measure. For $0 \leqslant t \leqslant T,$ the set $\mathcal{U}_t$ comprises the $\mathbb{F}-$progressively measurable stochastic processes $\left\{\nu_u\right\}_{t\leqslant u \leqslant T}$ satisfying 
$$
\mathbb{E}\left[ \int_t^T \nu_u^2\,du\right] < \infty.
$$
Moreover, given a Markovian multidimensional process $\{x_u\}_{t\leqslant u \leqslant T},$ we write
$$
\mathbb{E}_{t,x}\left[\cdot \right] = \mathbb{E}\left[\cdot | x_t = x \right]. 
$$
The letter $C$ denotes a generic positive constant that may change from line to line. It will depend on all model parameters, unless we explicitly state otherwise.

\section{The model} \label{sec:TheModel}

\subsection{State processes}

For $t\in \left[0,T\right],$ let us consider an agent trading shares of a certain asset, during the time horizon $\left[t,T\right],$ with turnover rate $\{\nu_u\}_{t\leqslant u \leqslant T},$ i.e., $\nu_u$ is the instantaneous rate at which this trader negotiates at time $u.$ Thus, intuitively speaking, if $\nu_u > 0$ (resp., $\nu_u < 0$), then the agent is instantaneously buying (resp., selling) shares of the asset at time $u.$ Her inventory process $\{Q_u\}_{t\leqslant u \leqslant T}$ has dynamics
\begin{equation} \label{eq:InventoryProcess}
    \begin{cases}
        dQ^{\nu}_u = \nu_u\,du, \\
        Q^{\nu}_t = q.
    \end{cases}
\end{equation}
In this setting, we will assume that the agent undergoes an instantaneous temporary price impact per share $\lambda=\left\{ \lambda_u \right\}_{t \leqslant u \leqslant T}$ which is not necessarily linear on the agent's turnover rate; instead, we stipulate that $\lambda$ is an appropriate power of the absolute value of it. More precisely, we assume it takes the form
\begin{equation} \label{eq:nonlinear_temp_price_impact}
    \lambda_u = \kappa_u |\nu_u|^\phi.
\end{equation}
where the parameter $\phi \in \left]0,1\right]$ is exogenously given. This is a reasonable modelling from an empirical viewpoint, see \cite{almgren2005direct}. Other relevant works considering this form for the temporary price impact are \cite{graewe2018smooth,leal2020learning}, see also \cite[Section 6.7]{cartea2015algorithmic}. Moreover, we consider a stock whose price process $\{S_u\}_{t \leqslant u \leqslant T}$ evolves according to
\begin{equation} \label{eq:StockPriceProcess}
    \begin{cases}
        dS^{\nu}_u = \sigma_u\,dB_u,\\
        S^{\nu}_t = S.
    \end{cases}
\end{equation}
Above, the process $\{\sigma_u\}_{t\leqslant u \leqslant T}$ is the absolute volatility of the asset price. Therefore, in view of \eqref{eq:nonlinear_temp_price_impact}, the execution price $\widehat{S}^{\nu}$ the agent obtains is 
$$
\widehat{S}^{\nu}_u = S^{\nu}_u - \lambda_u \sign(\nu_u).
$$
Thus, the agent's cash process has dynamics
\begin{equation} \label{eq:CashProcess}
    \begin{cases}
        dX^\nu_u = -\widehat{S}^{\nu}_u \nu_u\,du = -S^\nu_u\nu_u\,du - \kappa_u |\nu_u|^{1+\phi} \,du,\\
        X^\nu_t = x,
    \end{cases}
\end{equation}

The book value of the agent's cash plus inventory at time $t,$ which we refer to as her wealth, is
\begin{equation} \label{eq:WealthProcess}
    w^{\nu}_u := X^{\nu}_u + Q^{\nu}_u S^{\nu}_u \hspace{1.0cm} (t \leqslant u \leqslant T).
\end{equation}
Using It\^o's formula, it follows that
\begin{align} \label{eq:WealthRep}
  \begin{split}
    w^{\nu}_T =&\, w^{\nu}_t + \int_t^T\left\{-(S^{\nu}_u \nu_u + \kappa_u|\nu_u|^{1+\phi}) \, du + S^{\nu}_u \nu_u\, du + Q^{\nu}_u \sigma_u\,dB_u \right\} \\
    =&\, w^{\nu}_t - \int_t^T \kappa_u |\nu_u|^{1+\phi} \, du + \int_t^T \sigma_u Q^{\nu}_u \, dB_u .
  \end{split}
\end{align}

From here on, we assume (with slight abuse of notations)
$$
\kappa_u = \kappa(\boldsymbol{y}_u) \text{ and } \sigma_u = \sigma(\boldsymbol{y}_u),
$$
where $\kappa,\,\sigma : \mathbb{R}^d \rightarrow \left[0,\infty\right[,$ for a positive integer $d,$ and that $\{\boldsymbol{y}_u\}_{t \leqslant u \leqslant T}$ is a $d-$dimensional Markov diffusion. More precisely, we suppose that there are functions $\boldsymbol{\alpha} : \mathbb{R}^d \rightarrow \mathbb{R}^d,$ $\boldsymbol{\beta} : \mathbb{R}^d \rightarrow \mathbb{R}^{d\times m},$ such that
\begin{equation} \label{eq:MarkovDiffusions}
  \begin{cases}
    d\boldsymbol{y}_u = \boldsymbol{\alpha}(\boldsymbol{y}_u)\,du + \boldsymbol{\beta}(\boldsymbol{y}_u)\,d\boldsymbol{W}_u,\\
    \boldsymbol{y}_t = \boldsymbol{y}.
  \end{cases}
\end{equation}

Henceforth, we make the subsequent hypotheses on the functions introduced above:
\begin{itemize}
    \item[\textbf{(H1)}] The functions $\boldsymbol{\alpha}$ and $\boldsymbol{\beta}$ are Lipschitz continuous.
    \item[\textbf{(H2)}] Both $\kappa$ and $\sigma$ are continuous functions and there are $\underline{\kappa},\overline{\kappa},\overline{\sigma}>0,$ $\underline{\sigma} \geqslant 0,$ such that $\overline{\kappa} \geqslant \kappa \geqslant \underline{\kappa}$ and $\overline{\sigma} \geqslant \sigma \geqslant \underline{\sigma}.$
\end{itemize}

\subsection{Performance criteria and the value function: the regularized problem} \label{subsec:PerfCritReg}

In Section \ref{sec:AnalysisPDEs}, we consider the regularized problem, i.e., the circumstance in which we do not require strict execution but penalize non-vanishing terminal inventory holdings. We intend to work under the dynamic programming paradigm of stochastic optimal control, leading us to the following definition.
\begin{definition}
Given $t \in \left[0,T\right],$ our performance assessment of a strategy $\nu \in \mathcal{U}_t$ is made via the criterion
\begin{align} \label{eq:PerfCrit}
  \begin{split}
    J^{\nu}(t,x,S,q,\boldsymbol{y}) :=&\, \mathbb{E}_{t,x,S,q,\boldsymbol{y}}\left[w^{\nu}_T - \left( x+qS \right) \right] - \mathbb{E}_{t,x,S,q,\boldsymbol{y}}\left[A\left|Q^\nu_T\right|^{1+\phi} +  \gamma \int_t^T \sigma^{1+\phi}_u\left|Q^{\nu}_u\right|^{1+\phi}\,du \right]  \\
    =&\, \mathbb{E}_{t,q,\boldsymbol{y}}\left[\int_t^T \left\{-\kappa(\boldsymbol{y}_u)\left|\nu_u\right|^{1+\phi} -  \gamma \sigma^{1+\phi}(\boldsymbol{y}_u)\left|Q^\nu_u\right|^{1+\phi}\right\}\,du - A\left|Q^\nu_T\right|^{1+\phi} \right],
  \end{split}
\end{align}
where $A>0$ is a constant.
\end{definition}
\begin{remark}
After taking the supremum over the admissible strategies $\nu$, as \eqref{eq:PerfCrit} suggests, the value function will not depend on the state variables $x$ and $S.$ Thus, from now on, we will use the slight abuse of notation $J^\nu(t,x,S,q,\boldsymbol{y}) = J^\nu(t,q,\boldsymbol{y}).$
\end{remark}
\begin{remark}
The assumption $\phi \in \left]0,1\right]$ ensures that $J^\nu$ is well-defined, for each $t \in \left[0,T\right[$ and $\nu \in \mathcal{U}_t.$
\end{remark}

For a given $\nu \in \mathcal{U}_t$, the criterion $J^\nu$ defined in \eqref{eq:PerfCrit} includes two parts. The first one comprises the expectation of the difference between the agent's terminal wealth, $w_T,$ and her initial cash plus the pre-trade price, $x+qS.$ Therefore, we take an IS viewpoint. Two penalization terms constitute the remaining part of $J^\nu:$ (i) The term proportional to $\left|Q_T^\nu\right|^{1+\phi},$ for ending up with terminal inventory; (ii) The integral $\int_t^T \sigma^{1+\phi}\left(\boldsymbol{y}_u\right)\left|Q^\nu_u\right|^{1+\phi}\,du,$ which represents a sense of urgency of the trader.\footnote{We choose the power $1+\phi$ here to make the functional $J$ homogeneous in the inventory state variable $q$, as we will later show. Other works such as \cite{graewe2018smooth,leal2020learning} proceed with their investigations in the same way.} We observe from the identity \eqref{eq:WealthRep} that the addition of the latter term is a natural risk management tool to control $\left\{ \int_t^s \sigma_u Q^\nu_u\,dB_u \right\}_{t \leqslant s \leqslant T},$ which is a source of uncertainty in the terminal wealth $w_T$ collected by the agent via following her strategy. In particular, when $\phi = 1,$ this recovers the popular mean-variance framework of \cite{almgren1999value}. Furthermore, in view of the form of the expectation of $w^\nu_T - w^\nu_t,$ it seems appropriate to consider the power $1+\phi$ as we do here (for both terms in (i) and (ii) we described above). We will show that these modeling choices do lead us to a dimensionality reduction, viz., they allow us to drop the dependence on the variable $q,$ in a precise sense that we will discuss briefly. See the works \cite{graewe2018smooth,leal2020learning} for a similar approach to related problems.

The parameter $A$ in \eqref{eq:PerfCrit} makes the trader tolerant for finishing the schedule with a nonzero inventory. Mathematically, it has the effect of regularizing the problem. In Section \ref{sec:singular}, we will establish that, as $A$ tends to infinity, the optimal strategy of the regularized problem converges, in an appropriate sense, to the solution of the one in which complete execution is required. However, prior to taking limits, we do obtain estimates on the remaining terminal inventory in terms of $A.$ Thus, we indicate how large a trader should choose $A$ to guarantee the execution of a given percentage of her initial inventory. The interpretation of the parameter $\gamma$ in \eqref{eq:PerfCrit} is that it represents the risk aversion of the agent. In the linear temporary price impact case, in which case $\phi =1,$ then we identify $2\gamma$ as the risk aversion parameter for a constant absolute risk aversion model, see \cite{cartea2015algorithmic,gueant2016financial}. We notice that for the same level of risk aversion $\gamma,$ the trader is more (less) urgent for higher (lower) variance levels. We remark that we can treat other forms of stochastic urgency parameters using the techniques of the present work, under mild assumptions. For concreteness, we proceed with the model we presented above.

Alongside (\textbf{H1}) and (\textbf{H2}), we make the following hypothesis on $A:$
\begin{itemize}
    \item[\textbf{(H3)}] The terminal penalization parameter satisfies $A > \left(\frac{\gamma \overline{\sigma}^{1+\phi} \overline{\kappa}^{\frac{1}{\phi}}}{\phi} \right)^{\frac{\phi}{\phi + 1}}.$
\end{itemize}
We remark that (\textbf{H3}) is convenient (mainly for notational purposes), but it is not a necessary assumption. In addition, since our main interest is the regime where $A$ is large, it is not restrictive.

We subtract the quantity $x + qS $ in (\ref{eq:PerfCrit}) envisaging to attain a dimensionality reduction. In analogy to what is exposed in \cite{gueant2016financial}, this can be interpreted as a comparison between our revenue from following strategy $\nu,$ during the time window $\left[t,T\right],$ and the book value of initial inventory, $x+qS.$ Therefore, we follow the IS paradigm by considering these performance criteria. In the sequel, we introduce our value function.

\begin{definition}
The value function $J$ is given by
\begin{equation} \label{eq:ValueFn}
    J(t,q,\boldsymbol{y}) := \sup_{\nu \in \mathcal{U}_t} J^\nu(t,q,\boldsymbol{y}) \hspace{1.0cm}\left( (t,q,\boldsymbol{y}) \in \left[0,T\right]\times \mathbb{R} \times \mathbb{R}^d \right).
\end{equation}
\end{definition}

\subsection{Performance criteria and the value function: the singular problem} \label{subsec:PerfCritSing}

In Section \ref{sec:singular}, we will be concerned with the analysis of the constrained problem:
\begin{equation} \label{eq:StrictOptm}
    \sup_{\nu \in \mathcal{U}_c} \left( J_\infty^\nu(q,\,\boldsymbol{y}) := \mathbb{E}_{0,\,q,\,\boldsymbol{y}}\left[-\int_0^T \left\{\kappa(\boldsymbol{y}_t)\left|\nu_t\right|^{1+\phi} + \gamma \sigma^{1+\phi}(\boldsymbol{y}_t)\left| Q^\nu_t\right|^{1+\phi} \right\}\,dt \right] \right).
\end{equation}
We call the problem ``constrained'' because we define the set of admissible controls $\mathcal{U}_c,$ figuring above, as
\begin{equation} \label{eq:ConstrainedSet}
    \mathcal{U}_c := \left\{ \left\{ \nu_t\right\}_{0\leqslant t \leqslant T} \in \mathbb{L}^{1+\phi} : \int_0^T \nu_t\,dt = -q,\, \mathbb{P}-a.s. \right\},
\end{equation}
where 
$$
\mathbb{L}^{1+\phi} := \left\{ \left\{ \nu_t \right\}_{0\leqslant t \leqslant T} :  \left\{ \nu_t\right\}_t \text{ is } \mathbb{F}-\text{progressively measurable, and } \mathbb{E}\left[ \int_0^T \left|\nu_t\right|^{1+\phi}\,dt \right] < \infty \right\}.
$$
That is, we stipulate in \eqref{eq:ConstrainedSet} an execution constraint. The performance criteria for the current problem are the functionals $J_{\infty}^\nu$ we defined in \eqref{eq:StrictOptm}. We remark that $J_{\infty}^\nu(q,\,\boldsymbol{y}) = J^\nu(0,\,q,\,\boldsymbol{y}),$ for $\nu \in \mathcal{U}_c.$ Moreover, we notice that, properly identifying processes of $\mathbb{L}^{1+\phi}$ which agree $dt\times d\mathbb{P}-$a.e.a.s., we can render this set into a Banach space by endowing it with the norm
$$
\| \nu \|_{1+\phi} := \mathbb{E}\left[ \int_0^T \left|\nu_t\right|^{1+\phi}\,dt \right]^{\frac{1}{1+\phi}}.
$$
In view of the form of our performance criteria, the membership in $\mathbb{L}^{1+\phi}$ provides the natural integrability condition for a solution of the problem \eqref{eq:StrictOptm}. The other constraint we placed in the definition of $\mathcal{U}_c$ in \eqref{eq:ConstrainedSet} means precisely that we are only interested in strategies guaranteeing the complete execution of the initial inventory.

\section{Analysis of the regularized problem} \label{sec:AnalysisPDEs}

\subsection{The Hamilton-Jacobi-Bellman equation}

From \cite[Theorem 4.3.1]{pham2009continuous}, we know that the value function $J$ is a viscosity solution of the Hamilton-Jacobi-Bellman (HJB) equation
$$
\partial_t J + \mathcal{L} J + \sup_{\nu}\left\{-\kappa\left|\nu\right|^{1+\phi} +\nu\partial_q J \right\} - \gamma \sigma^{1+\phi} \left|q\right|^{1+\phi} = 0,
$$
with $J(T,q,\boldsymbol{y}) = -A \left| q \right|^{1+\phi},$ where $\mathcal{L}$ is the infinitesimal generator of $\{\boldsymbol{y}_t\}_t,$
$$
\begin{cases}
    \mathcal{L} := \frac{1}{2}\tr\left( \boldsymbol{\beta} \boldsymbol{\beta}^\intercal (\boldsymbol{y}) D_{\boldsymbol{y}}^2 \right) + \boldsymbol{\alpha}(\boldsymbol{y})\cdot D_{\boldsymbol{y}},\\
    \text{for } \left( D_{\boldsymbol{y}} \right)_i := \partial_{y_i} \text{ and } \left( D_{\boldsymbol{y}} \right)_{ij} := \partial_{y_i} \partial_{y_j},\, i,j \in \left\{1,\ldots,d \right\}.
\end{cases}
$$
We have 
$$
\sup_{\nu}\left\{-\kappa\left|\nu\right|^{1+\phi} +\nu\partial_q J \right\} = \kappa \widetilde{H}\left(\frac{\partial_q J }{\kappa} \right),
$$
where
$$
 \widetilde{H}(p) = \phi\left( \frac{|p|}{1+\phi} \right)^{1+\frac{1}{\phi}},
$$
and the optimal control in feedback form is
\begin{equation} \label{eq:ControlFeedback1}
    \nu^*(t,q,\boldsymbol{y}) := \widetilde{H}^\prime\left( \frac{\partial_q J(t,q,\boldsymbol{y})}{\kappa(\boldsymbol{y})} \right).
\end{equation}
In this way, the HJB reads
\begin{equation} \label{eq:MainHJB}
    \partial_t J + \mathcal{L} J + \kappa \widetilde{H}\left(\frac{\partial_q J }{\kappa} \right) -  \gamma \sigma^{1+\phi}\left|q\right|^{1+\phi} = 0.
\end{equation}
We propose the ansatz 
\begin{equation} \label{eq:Ansatz}
    J(t,q,\boldsymbol{y}) = z(t,\boldsymbol{y})\left| q \right|^{1+\phi},
\end{equation}
which leads to the PDE
\begin{equation} \label{eq:ActualMainPDE}
    \begin{cases}
        \partial_t z + \mathcal{L}z + \kappa H\left( \frac{z}{\kappa} \right) - \gamma \sigma^{1+\phi} = 0, \\
        z(T,\boldsymbol{y}) = -A,
    \end{cases}
\end{equation}
where $H(p) = (1+\phi)^{1+\frac{1}{\phi}}\widetilde{H}(p) = \phi |p|^{1+\frac{1}{\phi}}.$ Arguing as in \cite[Lemma 2.7]{graewe2018smooth}, we can show that $z$ solves \eqref{eq:ActualMainPDE} if, and only if, $\left( t,\,q,\,\boldsymbol{y} \right) \mapsto z\left(t,\boldsymbol{y}\right)|q|^{1+\phi}$ solves \eqref{eq:MainHJB}. In the next Section, we turn to the analysis of \eqref{eq:ActualMainPDE}. In particular, we will show a verification result, guaranteeing that \eqref{eq:Ansatz} holds.

The most fundamental aspect of the analysis of the regularized problem is the investigation of the PDE \eqref{eq:ActualMainPDE}. Thus, we refer to this equation as our main PDE. The present Section is devoted to establishing the existence and uniqueness of a continuous and bounded viscosity solution of it.

\subsection{Some previous results} \label{subsec:HypothesesAndPrevResults}

The subsequent theorem will be key in the remainder of this section. It is a particular case of \cite[Theorem 3.42]{pardoux2014stochastic}.
\begin{theorem} \label{thm:Viscosity}
Under \textbf{(H1)} and \textbf{(H2)}, let $c,f:\left[0,T\right]\times\mathbb{R}^d \rightarrow \mathbb{R}$ and $g : \mathbb{R}^d \rightarrow \mathbb{R}$ be three continuous and bounded functions. Define
\begin{equation} \label{eq:FeynmanKac}
    h(t,\boldsymbol{y}) := \mathbb{E}_{t,\boldsymbol{y}}\left[\int_t^T e^{\int_t^u c(\tau,\boldsymbol{y}_\tau)\,d\tau}f(u,\boldsymbol{y}_u)\,du + e^{\int_t^T c(\tau,\boldsymbol{y}_\tau)\,d\tau}g(\boldsymbol{y}_T) \right].
\end{equation}
Then $h$ is continuous, and it is the unique viscosity solution of the PDE
\begin{equation} \label{eq:LinearPDE}
    \begin{cases}
        \partial_t h + \mathcal{L}h + ch + f = 0 &\text{ in } \left]0,T\right[\times \mathbb{R}^d, \\
        h|_{t=T} = g &\text{ in } \mathbb{R}^d,
    \end{cases}
\end{equation}
within the class $\mathcal{G}$ consisting of continuous functions satisfying
$$
\lim_{|\boldsymbol{y}|\rightarrow \infty} h(t,\boldsymbol{y})e^{-\delta\left[\log \left|\boldsymbol{y}\right|  \right]^2} = 0,
$$
for some $\delta > 0.$
\end{theorem}

In effect, from the representation (\ref{eq:FeynmanKac}), the property of comparison holds. We state it in Corollary \ref{cor:Comparison}.

\begin{corollary} \label{cor:Comparison}
Let $c,f,\widetilde{f} : \left[0,T\right]\times \mathbb{R}^d \rightarrow \mathbb{R},$ $g,\widetilde{g} : \mathbb{R}^d \rightarrow \mathbb{R}$ be five bounded continuous functions. Define $h$ as in (\ref{eq:FeynmanKac}), and likewise $\widetilde{h},$ the latter having $\widetilde{f}$ and $\widetilde{g}$ in place of $f$ and $g,$ respectively. If $f \geqslant \widetilde{f}$ and $g\geqslant \widetilde{g},$ then $h \geqslant \widetilde{h}.$
\end{corollary}

\subsection{Finding a subsolution and a supersolution} \label{subsec:SubSuper}

Let us introduce the operators
$$
\begin{cases}
\mathcal{A}(h) := \partial_t h + \mathcal{L}h + \kappa H\left( h / \kappa \right) - \gamma \sigma^{1+\phi} = \partial_t h + \mathcal{L}h + \phi \kappa^{-\frac{1}{\phi}}|h|^{1+\frac{1}{\phi}} -  \gamma \sigma^{1+\phi},\\
\mathcal{\underline{A}}(h) := \partial_t h + \mathcal{L}h + \phi\overline{\kappa}^{-\frac{1}{\phi}}|h|^{1+\frac{1}{\phi}} -  \gamma\overline{\sigma}^{1+\phi} , \\
\overline{\mathcal{A}}(h) := \partial_t h + \mathcal{L}h + \phi \underline{\kappa}^{-\frac{1}{\phi}}|h|^{1+\frac{1}{\phi}} -  \gamma\underline{\sigma}^{1+\phi} .
\end{cases}
$$
We observe that 
\begin{equation} \label{eq:OrdOperators}
    \overline{\mathcal{A}}(h) \geqslant \mathcal{A}(h) \geqslant \underline{\mathcal{A}}(h).
\end{equation}

We will use the above operators to build a subsolution and a supersolution to \eqref{eq:ActualMainPDE}, which will help us to prove the well-posedness of the latter PDE. These constructions will rely upon the following result on a class of Ordinary Differential Equations (ODEs).



We observe that, for each $a,b > 0$ and $r > 1$ subject to $b A^r - a > 0,$ the scalar initial value problem
$$
\begin{cases}
y^\prime = a - b|y|^r,\\
y(T) = -A,
\end{cases}
$$
admits a unique classical solution $y$ on $\left[0,T\right].$ Moreover, $y$ is monotone decreasing and 
\begin{equation} \label{eq:Inequalities}
    -A \leqslant y < -(a/b)^{1/r}.
\end{equation}
In effect, $y$ is given by $y(t) = F^{-1}(T-t),$ for the bijective differentiable mapping
$$
F : \xi \in \left[-A, -\left(a/b\right)^{1/r} \right[ \mapsto -\int_{-A}^\xi \frac{du}{a-b|u|^r} \in \left[0,\infty\right[,
$$
which satisfies $F^\prime > 0.$ Thus, choosing $a$ and $b$ suitably, and $r := 1+ 1/\phi,$ we infer that there are two differentiable deterministic functions $\underline{z},\,\overline{z}:\left[0,T\right] \rightarrow \left[-A,0\right[$ (independent of the state variable $\boldsymbol{y} \in \mathbb{R}^d$) solving
\begin{equation} \label{eq:SubSolution}
  \begin{cases}    
    \mathcal{\underline{A}}\left( \underline{z} \right) = 0, \\
    \underline{z}(T) = -A,
  \end{cases}
\end{equation}
and
\begin{equation} \label{eq:SuperSolution}
  \begin{cases}
    \overline{\mathcal{A}}\left(\overline{z}\right) =0,\\
    \overline{z}(T) = -A.
  \end{cases}
\end{equation}
Furthermore, they are subject to the bounds
\begin{equation} \label{eq:BoundsOnSubSuper}
    -A \leqslant \underline{z} \leqslant \overline{z} \leqslant -\left(\frac{\gamma \underline{\sigma}^{1+\phi} \underline{\kappa}^{\frac{1}{\phi}}}{\phi} \right)^{\frac{\phi}{\phi + 1}} \leqslant 0.
\end{equation}
The inequalities in the extremes of \eqref{eq:BoundsOnSubSuper} are straightforward to derive from \eqref{eq:Inequalities}, and we can show the one in the middle by standard ODE comparison arguments. We can also show that there exists a positive constant $C,$ independent of $A,$ such that
\begin{equation} \label{eq:SubAndSuperBlowupRate}
    \frac{1}{C(A^{-\frac{1}{\phi}} + T-t)^\phi} \leqslant |\underline{z}(t)| \leqslant \frac{C}{(A^{-\frac{1}{\phi}} + T-t)^\phi} \text{ and } \frac{1}{C(A^{-\frac{1}{\phi}} + T-t)^\phi} \leqslant |\overline{z}(t)| \leqslant \frac{C}{(A^{-\frac{1}{\phi}} + T-t)^\phi},
\end{equation}
for $0 \leqslant t \leqslant T,$ see \ref{app:BlowupRateOfSubAndSuper}. In the next Subsection, we will find the solution $z$ of \eqref{eq:ActualMainPDE}, in an appropriate sense, subject to $\underline{z} \leqslant z \leqslant \overline{z}.$ Intuitively, this is coherent with a comparison principle, cf. \eqref{eq:OrdOperators}.

\subsection{Existence and uniqueness properties of the main PDE}

We obtain existence and uniqueness results for (\ref{eq:ActualMainPDE}) through an iterative monotonicity method. For a description of this approach in other contexts, we refer to \cite[Chapter 7]{pao2012nonlinear} and \cite[Chapter 12]{wu2006elliptic}. Here, we apply this technique in the setting of viscosity solutions with milder hypotheses on model coefficients.

The first step is to define the bounded continuous coefficient $c,$ 
$$
c: = -\left( \phi + 1 \right) \left( \frac{\left| \underline{z} \right|}{\kappa} \right)^{\frac{1}{\phi}} = -\left( \phi + 1 \right) \left( - \frac{ \underline{z} }{\kappa} \right)^{\frac{1}{\phi}},
$$
and designate by $\mathcal{A}_c$ the operator
$$
\mathcal{A}_c : h \mapsto \partial_th + \mathcal{L}h + ch.
$$
Instead of solving (\ref{eq:ActualMainPDE}), we will solve the equivalent problem
\begin{equation} \label{eq:EquivMainPDE}
  \begin{cases}
    \mathcal{A}_c z + \kappa H\left( z/ \kappa \right) - \gamma \sigma^{1+\phi} - cz = 0 &\text{ in } \left]0,T\right[\times \mathbb{R}^d,\\
    z|_{t=T} = -A &\text{ in } \mathbb{R}^d.
  \end{cases}
\end{equation}

\begin{lemma} \label{lem:Step1}
Let $\underline{z}^{(1)},\, \overline{z}^{(1)} \in \mathcal{G}$ be the viscosity solutions of the PDEs
$$
\begin{cases}
    \mathcal{A}_c\underline{z}^{(1)} +  \underline{f}^{(1)} = 0  &\text{ in } \left]0,T\right[\times \mathbb{R}^d,\\
    \underline{z}^{(1)}|_{t=T} = -A &\text{ in } \mathbb{R}^d,
\end{cases}
$$
$$
\begin{cases}
    \mathcal{A}_c\overline{z}^{(1)} + \overline{f}^{(1)} = 0  &\text{ in } \left]0,T\right[\times\mathbb{R}^d, \\
    \overline{z}^{(1)}|_{t=T} = -A &\text{ in } \mathbb{R}^d,
\end{cases}
$$
where
$$
\underline{f}^{(1)} := - \gamma \sigma^{1+\phi} + \phi \kappa^{-\frac{1}{\phi}}\left| \underline{z} \right|^{1+\frac{1}{\phi}} - c\underline{z}.
$$
and
$$
\overline{f}^{(1)} := - \gamma \sigma^{1+\phi} + \phi \kappa^{-\frac{1}{\phi}}\left|\overline{z}\right|^{1+\frac{1}{\phi}} - c\overline{z}
$$
Then,
\begin{equation} \label{eq:StepOneViscMainPDE}
    \underline{z} \leqslant \underline{z}^{(1)} \leqslant \overline{z}^{(1)} \leqslant \overline{z}.
\end{equation}
\end{lemma}
\begin{proof}
We notice that the functions $\underline{z}$ and $\overline{z}$ (see Section \ref{subsec:SubSuper}) solve
$$
\begin{cases}
    \mathcal{A}_c \underline{z} + \underline{f}= 0,\\
    \mathcal{A}_c \overline{z} + \overline{f} = 0, \\
    \underline{z}|_{t=T} = -A = \overline{z}|_{t=T},
\end{cases}
$$
where
$$
\begin{cases}
    \underline{f} := - \gamma \overline{\sigma}^{1+\phi} + \phi \overline{\kappa}^{-\frac{1}{\phi}}\left| \underline{z} \right|^{1+\frac{1}{\phi}} - c\underline{z},\\
    \overline{f} := - \gamma\underline{\sigma}^{1+\phi} + \phi \underline{\kappa}^{-\frac{1}{\phi}}\left|\overline{z}\right|^{1+\frac{1}{\phi}} - c\overline{z}.
\end{cases}
$$
From the relations
$$
\underline{f}^{(1)} - \underline{f} =  \gamma (\overline{\sigma}^{1+\phi} - \sigma^{1+\phi}) + \phi\left( \kappa^{-\frac{1}{\phi}} - \overline{\kappa}^{-\frac{1}{\phi}} \right) \left| \underline{z}\right|^{1+\frac{1}{\phi}} \geqslant 0,
$$
and
$$
\underline{z}^{(1)}|_{t=T} = \overline{z}^{(1)}|_{t=T},
$$
we conclude through Corollary \ref{cor:Comparison} that
$$
\underline{z} \leqslant \underline{z}^{(1)}.
$$
Likewise, from the inequalities
$$
\overline{f}^{(1)} - \overline{f} = \gamma (\underline{\sigma}^{1+\phi} - \sigma^{1+\phi}) + \phi \left( \kappa^{-\frac{1}{\phi}} - \underline{\kappa}^{-\frac{1}{\phi}}\right)\left| \overline{z} \right|^{1+\frac{1}{\phi}} \leqslant 0
$$
and
\begin{align*}
    \overline{f}^{(1)} - \underline{f}^{(1)} &= \phi \kappa^{-\frac{1}{\phi}} \left( \left| \overline{z} \right|^{1+\frac{1}{\phi}} - \left| \underline{z} \right|^{1+\frac{1}{\phi}} \right) - c \left( \overline{z} - \underline{z} \right) \\
    &\geqslant \left[ - (\phi + 1) \kappa^{-\frac{1}{\phi}} \left| \underline{z} \right|^{\frac{1}{\phi}} - c \right] \left(\overline{z} - \underline{z}\right) \\
    &= 0,
\end{align*}
alongside the fact that $\underline{z}|_{t=T} = \overline{z}|_{t=T},$ we deduce that we can apply Corollary \ref{cor:Comparison} to deduce the other two inequalities in (\ref{eq:StepOneViscMainPDE}).
\end{proof}

\begin{lemma} \label{lem:Step2}
We set $\underline{z}^{(0)} := \underline{z}$ and $\overline{z}^{(0)} := \overline{z}.$ For some $k\geqslant 1,$ we assume that there are functions $\left\{\underline{z}^{(l)},\, \overline{z}^{(l)}\right\}_{l=0}^k \subseteq \mathcal{G}$ solving the PDEs 
$$
\begin{cases}
    \mathcal{A}_c \underline{z}^{(l)} + \underline{f}^{(l)} = 0 &\text{ in } \left]0,T\right[\times \mathbb{R}^d,\\
    \text{where } \underline{f}^{(l)} := - \gamma \sigma^{1+\phi} + \phi \kappa^{-\frac{1}{\phi}}\left|\underline{z}^{(l-1)}\right|^{1+\frac{1}{\phi}} - c\underline{z}^{(l-1)},\\
    \underline{z}^{(l)} = -A &\text{ in } \mathbb{R}^d,
\end{cases}
$$
$$
\begin{cases}
    \mathcal{A}_c \overline{z}^{(l)} + \overline{f}^{(l)} = 0 &\text{ in } \left]0,T\right[\times \mathbb{R}^d,\\
    \text{where } \overline{f}^{(l)} := - \gamma \sigma^{1+\phi} + \phi \kappa^{-\frac{1}{\phi}}\left|\overline{z}^{(l-1)}\right|^{1+\frac{1}{\phi}} - c\overline{z}^{(l-1)},\\
    \overline{z}^{(l)} = -A &\text{ in } \mathbb{R}^d,
\end{cases}
$$
in the viscosity sense, for $1 \leqslant l \leqslant k,$ and satisfying
$$
\underline{z} = \underline{z}^{(0)} \leqslant \cdots \leqslant \underline{z}^{(k-1)}\leqslant \underline{z}^{(k)} \leqslant \overline{z}^{(k)} \leqslant \overline{z}^{(k-1)} \leqslant \cdots \leqslant \overline{z}^{(0)} = \overline{z}.
$$
Then, considering the viscosity solutions of $\underline{z}^{(k+1)},\, \overline{z}^{(k+1)} \in \mathcal{G}$ of
$$
\begin{cases}
    \mathcal{A}_c\underline{z}^{(k+1)} + \underline{f}^{(k+1)} = 0, &\text{ in } \left]0,T\right[\times\mathbb{R}^d, \\
    \text{where } \underline{f}^{(k+1)} := - \gamma \sigma^{1+\phi} + \phi \kappa^{-\frac{1}{\phi}}\left|\underline{z}^{(k)}\right|^{1+\frac{1}{\phi}} - c\underline{z}^{(k)},\\
    \underline{z}^{(k+1)}|_{t=T}= -A &\text{ in } \mathbb{R}^d,
\end{cases}
$$
and
$$
\begin{cases}
    \mathcal{A}_c\overline{z}^{(k+1)} + \overline{f}^{(k+1)} = 0  &\text{ in } \left]0,T\right[\times\mathbb{R}^d, \\
    \text{where } \overline{f}^{(k+1)} := - \gamma \sigma^{1+\phi} + \phi \kappa^{-\frac{1}{\phi}}\left|\overline{z}^{(k)}\right|^{1+\frac{1}{\phi}} - c\underline{z}^{(k)},\\
    \overline{z}^{(k+1)}|_{t=T} = -A &\text{ in } \mathbb{R}^d,
\end{cases}
$$
we have
$$
\underline{z}^{(k)} \leqslant \underline{z}^{(k+1)} \leqslant  \overline{z}^{(k+1)} \leqslant \overline{z}^{(k)} .
$$
\end{lemma}
\begin{proof}
Under the present assumptions, we have $\overline{z}^{(l)} \geqslant \underline{z}$ and $\underline{z}^{(l)} \geqslant \underline{z},$ for all $0 \leqslant l \leqslant k.$ Hence, we can estimate
\begin{align*}
    \underline{f}^{(k+1)} - \underline{f}^{(k)} &= \phi \kappa^{-\frac{1}{\phi}} \left( \left|\underline{z}^{(k)}\right|^{1+\frac{1}{\phi}} - \left|\underline{z}^{(k-1)}\right|^{1+\frac{1}{\phi}} \right) - c(\underline{z}^{(k)} - \underline{z}^{(k-1)}) \\
    &\geqslant \left[ -\left(\phi + 1 \right)\kappa^{-\frac{1}{\phi}}\left| \underline{z}^{(k-1)} \right|^{\frac{1}{\phi}} - c\right](\underline{z}^{(k)} - \underline{z}^{(k-1)})\\
    &\geqslant 0.
\end{align*}
Likewise, we show
$$
\overline{f}^{(k+1)} - \underline{f}^{(k+1)} \geqslant 0,
$$
as well as
$$
\overline{f}^{(k+1)} - \overline{f}^{(k)} \leqslant 0.
$$
Since $\underline{z}^{(k+1)}|_{t=T} = \underline{z}^{(k)}|_{t=T} = \overline{z}^{(k)}|_{t=T} = \overline{z}^{(k+1)}|_{t=T},$ we conclude the desired result from Corollary \ref{cor:Comparison}. 
\end{proof}

From Lemmas \ref{lem:Step1} and \ref{lem:Step2}, we conclude the well-definiteness of the sequences $\{\underline{z}^{(k)}\}_{k\geqslant 0}$ and $\{\overline{z}^{(k)}\}_{k\geqslant 0}$ such that they are viscosity solutions in the class $\mathcal{G}$ of the PDEs
$$
\begin{cases}
\mathcal{A}_c \underline{z}^{(k)} =  \gamma \sigma^{1+\phi} - \phi \kappa^{-\frac{1}{\phi}}\left| \underline{z}^{(k-1)} \right|^{1+\frac{1}{\phi}} + c \underline{z}^{(k-1)} &\text{ in } \left]0,T\right[\times \mathbb{R}^d,\\
\underline{z}^{(k)} = -A &\text{ in } \mathbb{R}^d,
\end{cases}
$$
$$
\begin{cases}
\mathcal{A}_c \overline{z}^{(k)} =  \gamma \sigma^{1+\phi} - \phi \kappa^{-\frac{1}{\phi}}\left| \overline{z}^{(k-1)} \right|^{1+\frac{1}{\phi}} + c \overline{z}^{(k-1)} &\text{ in } \left]0,T\right[\times \mathbb{R}^d,\\
\overline{z}^{(k)} = -A &\text{ in } \mathbb{R}^d,
\end{cases}
$$
for $k\geqslant 1,$ and that satisfy 
$$ \underline{z}^{(0)} := \underline{z} \text{ and } \overline{z}^{(0)} := \overline{z},$$ 
for $k=0$. We emphasize that, in particular, the membership in the class $\mathcal{G}$ ensures their continuity. Furthermore, they satisfy
$$
\underline{z} \leqslant \underline{z}^{(k)} \leqslant \underline{z}^{(k+1)} \leqslant \overline{z}^{(k+1)} \leqslant\overline{z}^{(k)} \leqslant \overline{z},
$$
for all $k\geqslant 0.$ Therefore, it is licit to define the following pointwise limits
$$
\underline{z}^*(t,\boldsymbol{y}) := \lim_{k\rightarrow \infty}\underline{z}^{(k)}(t,\boldsymbol{y}) \text{ and } \overline{z}^*(t,\boldsymbol{y}) := \lim_{k\rightarrow \infty}\overline{z}^{(k)}(t,\boldsymbol{y}) \hspace{1.0cm} ((t,\boldsymbol{y}) \in \left[0,T\right] \times \mathbb{R}^d).
$$
We observe that they satisfy
$$
\underline{z} \leqslant \underline{z}^* \leqslant \overline{z}^* \leqslant \overline{z}.
$$

\begin{theorem} \label{thm:MainPDE}
The PDE (\ref{eq:EquivMainPDE}) has a unique bounded continuous viscosity solution $z.$ Moreover, it is given by $z = \underline{z}^* = \overline{z}^*.$ 
\end{theorem}
\begin{proof}
Firstly, we note that Theorem \ref{thm:Viscosity} implies
\begin{align*}
    \begin{split}
        \underline{z}^{(k)}(t,\boldsymbol{y}) = \mathbb{E}_{t,\boldsymbol{y}}\Bigg[\int_t^T e^{\int_t^u c(\tau,\boldsymbol{y}_\tau)\,d\tau}\Bigg\{&- \gamma \sigma^{1+\phi}(\boldsymbol{y}_u) + \phi \kappa(\boldsymbol{y}_u)^{-\frac{1}{\phi}} \left| \underline{z}^{(k-1)}(u,\boldsymbol{y}_u) \right|^{1+\frac{1}{\phi}}\\
        &- c(u,\boldsymbol{y}_u)\underline{z}^{(k-1)}(u,\boldsymbol{y}_u) \Bigg\}\,du - A e^{\int_t^T c(u,\boldsymbol{y}_u)\,du} \Bigg].
    \end{split}
\end{align*}
Next, we can let $k\rightarrow \infty$ and use the Dominated Convergence Theorem to deduce that $\underline{z}^*$ solves
\begin{align} \label{eq:FixedPoint}
    \begin{split}
        \underline{z}^*(t,\boldsymbol{y}) = \mathbb{E}_{t,\boldsymbol{y}}\Bigg[\int_t^T e^{\int_t^u c(\tau,\boldsymbol{y}_\tau)\,d\tau}\Bigg\{&- \gamma \sigma^{1+\phi}(\boldsymbol{y}_u) + \phi \kappa(\boldsymbol{y}_u)^{-\frac{1}{\phi}} \left| \underline{z}^*(u,\boldsymbol{y}_u) \right|^{1+\frac{1}{\phi}}\\
        &- c(u,\boldsymbol{y}_u)\underline{z}^*(u,\boldsymbol{y}_u) \Bigg\}\,du - A e^{\int_t^T c(u,\boldsymbol{y}_u)\,du} \Bigg].
    \end{split}
\end{align}
From the representation \eqref{eq:FixedPoint}, we can show that $\underline{z}^*$ is continuous, see \ref{app:ContinuityIssue}. Therefore, according to Theorem \ref{thm:Viscosity}, the function on the right-hand side of (\ref{eq:FixedPoint}), which we proved to be equal to $\underline{z}^*,$ is also continuous and solves the PDE
$$
\begin{cases}
\mathcal{A}_c(\underline{z}^*) =  \gamma \sigma^{1+\phi} - \phi \kappa^{-\frac{1}{\phi}}\left| \underline{z}^* \right|^{1+\frac{1}{\phi}} + c\underline{z}^* &\text{ in } \left]0,T\right[\times\mathbb{R}^d, \\
\underline{z}^*|_{t=T} = -A &\text{ in } \mathbb{R}^d.
\end{cases}
$$
in the viscosity sense. In other words, $\underline{z}^*$ is a viscosity solution of (\ref{eq:EquivMainPDE}) or, equivalently, this function solves (\ref{eq:ActualMainPDE}). We can make the same argument to show that $\overline{z}^*$ enjoys this same property. This proves the existence part of the Theorem. The fact that $\underline{z}^* = \overline{z}^*$ will follow from the proof of the uniqueness of continuous and bounded viscosity solutions of \eqref{eq:ActualMainPDE}, which we now turn to show.

Let us assume $\widetilde{z}_i,$ $i=1,2,$ are two bounded continuous viscosity solutions of \eqref{eq:ActualMainPDE}. Applying Theorem \ref{thm:Viscosity}, we infer
$$
\widetilde{z}_i(t,\boldsymbol{y}) = \mathbb{E}_{t,\boldsymbol{y}}\left[\int_t^T\left\{- \gamma \sigma^{1+\phi}(\boldsymbol{y}_u) + \phi \kappa(\boldsymbol{y}_u)^{-\frac{1}{\phi}}\left| \widetilde{z}_i(u,\boldsymbol{y}_u) \right|^{1+\frac{1}{\phi}} \right\}\,du \right] -A.
$$
Setting $\delta := \widetilde{z}_1 - \widetilde{z}_2,$ we obtain
\begin{equation} \label{eq:Uniqueness1}
    \delta(t,\boldsymbol{y}) = \mathbb{E}_{t,\boldsymbol{y}}\left[\int_t^T g(u,\boldsymbol{y}_u)\delta(u,\boldsymbol{y}_u)\,du \right],
\end{equation}
where
$$
g(t,\boldsymbol{y}) := \begin{cases}
\phi \kappa(\boldsymbol{y})^{-\frac{1}{\phi}}\left( \frac{ \left| \widetilde{z}_1(u,\boldsymbol{y}) \right|^{1+\frac{1}{\phi}} - \left| \widetilde{z}_2(u,\boldsymbol{y} ) \right|^{1+\frac{1}{\phi}} }{ \widetilde{z}_1(u,\boldsymbol{y} ) - \widetilde{z}_2(u,\boldsymbol{y} )} \right) &\text{ if } \widetilde{z}_1(t,\boldsymbol{y}) \neq \widetilde{z}_2(t,\boldsymbol{y}),\\
(\phi + 1) \kappa(\boldsymbol{y} )^{-\frac{1}{\phi}} \left| \widetilde{z}_1(t,\boldsymbol{y}) \right|^{\frac{1}{\phi}}\sign\left(\widetilde{z}_1(t,\boldsymbol{y})\right)  &\text{ otherwise.}
\end{cases}
$$
We notice that $g$ is bounded. Let $C>0$ be a constant such that $|g| \leqslant C.$ We set
$$
\Delta(t) := \sup_{\boldsymbol{y}}|\delta(t,\,\boldsymbol{y})|.
$$
Therefore, from identity (\ref{eq:Uniqueness1}) we infer
\begin{equation} \label{eq:Uniqueness3}
    \Delta(t) \leqslant C\int_t^T \Delta(u)\,du \hspace{1.0cm} (0 \leqslant t \leqslant T).
\end{equation}
An application of Gronwall's Lemma gives $\Delta \equiv 0,$ whence $\widetilde{z}_1 \equiv \widetilde{z}_2.$ This finishes the proof of the Theorem.
\end{proof}

\begin{corollary}
The convergences $\lim_{k\rightarrow \infty} \overline{z}^{(k)} = \lim_{k\rightarrow \infty} \underline{z}^{(k)} = z$ are uniform over compact subsets of $\left[0,T\right[\times \mathbb{R}^d.$
\end{corollary}
\begin{proof}
From Theorem \ref{thm:MainPDE}, we know that the limiting functions $\overline{z}^*$ and $\underline{z}^*$ indeed coincide and are continuous. An application of Dini's Theorem, see \cite[Theorem 7.13]{rudin1964principles}, gives the result we stated.
\end{proof}

\subsection{A verification result and some properties of the optimal strategy} \label{eq:propertiesOptimalStrat}

The first result we expose in this subsection are representations of the optimal speed of trading and inventory in terms of the solution of (\ref{eq:ActualMainPDE}).

\begin{theorem} \label{thm:OptimalProcs}
The value function is indeed given by \eqref{eq:Ansatz}. Thus, the optimal speed of trading $\{\nu^*_t\}_{0\leqslant t \leqslant T}$ and the corresponding optimal inventory holdings $\left\{Q^*_t := Q^{\nu^*}_t\right\}_{0\leqslant t \leqslant T}$ are given by
\begin{equation} \label{eq:OptimalSpeed}
    \nu^*_t = - Q_0 \left( - \frac{z(t,\boldsymbol{y}_t)}{ \kappa(\boldsymbol{y}_t)} \right)^{\frac{1}{\phi}} \exp\left( - \int_0^t \left( - \frac{z(u,\boldsymbol{y}_u)}{ \kappa(\boldsymbol{y}_u)} \right)^{\frac{1}{\phi}} \, du \right),
\end{equation}
and
\begin{equation} \label{eq:OptimalInventory}
    Q^*_t = Q_0 \exp\left( - \int_0^t \left( - \frac{z(u,\boldsymbol{y}_u)}{ \kappa(\boldsymbol{y}_u)} \right)^{\frac{1}{\phi}} \, du \right).
\end{equation}
\end{theorem}
\begin{proof}
We can prove the verification result, namely, that $J(t,q,\boldsymbol{y}) = z(t,\boldsymbol{y})|q|^{1+\phi},$ just as in \cite[Proposition 2.10]{graewe2018smooth} or \cite[Proposition 2.8]{horst2020continuous}. The proof here is even easier because we are considering the regularized problem, and there is no jump component on the trader's inventory process; for completeness, we include it in \ref{app:Verification}.

From \eqref{eq:ControlFeedback1} and \eqref{eq:Ansatz}, and recalling that $z \leqslant \overline{z} \leqslant 0,$ we derive
\begin{equation} \label{eq:ControlFeedback2}
    \nu^*(t,q,\boldsymbol{y}) := \sign\left( \partial_q J(t,q,\boldsymbol{y}) \right) \left( \frac{|\partial_q J(t,q,\boldsymbol{y})|}{(1+\phi)\kappa(\boldsymbol{y})} \right)^{\frac{1}{\phi}} = - \left( - \frac{z(t,\boldsymbol{y})}{ \kappa(\boldsymbol{y})} \right)^{\frac{1}{\phi}}q.
\end{equation}
Therefore, from 
$$
dQ^*_t = \nu^*\left( t,\, Q^*_t,\, \boldsymbol{y}_t \right) \,dt = - \left( - \frac{z(t,\boldsymbol{y}_t)}{ \kappa(\boldsymbol{y}_t)} \right)^{\frac{1}{\phi}} Q^*_t \, dt
$$
we deduce
\begin{equation*}
    Q^*_t = Q_0 \exp\left( - \int_0^t \left( - \frac{z(u,\boldsymbol{y}_u)}{ \kappa(\boldsymbol{y}_u)} \right)^{\frac{1}{\phi}} \, du \right).
\end{equation*}
This proves \eqref{eq:OptimalInventory}. Using \eqref{eq:OptimalInventory} in \eqref{eq:ControlFeedback2}, we show \eqref{eq:OptimalSpeed}.
\end{proof}

\begin{corollary}
The optimal terminal inventory holdings satisfies
\begin{equation} \label{eq:InvBound}
    |Q^*_t| \leqslant |Q_0| \left( \frac{T  - t + A^{-1/\phi} }{T + A^{-1/\phi}} \right)^{\left( \frac{\ell}{\overline{\kappa}}\right)^{1/\phi}} ,
\end{equation}
for each $0 \leqslant t \leqslant T,$ where we have written
\begin{equation} \label{eq:LowerBoundCForInv}
    \ell := \inf_{ s \in \left[0,T\right] } \left[ |\overline{z}\left(s\right)|(T-s+A^{-1/\phi})^{\phi} \right].
\end{equation}
\end{corollary}
\begin{remark}
In view of \eqref{eq:BoundsOnSubSuper}, we observe that the constant $\ell$ we defined in \eqref{eq:LowerBoundCForInv} is strictly positive, and that it is independent of $A.$
\end{remark}
\begin{proof}
From \eqref{eq:OptimalInventory}, we deduce
\begin{align}
    \begin{split}
        |Q_t^*| &\leqslant |Q_0| \exp\left( - \frac{1}{\overline{\kappa}^{1/\phi}}\int_0^t |z(u,\boldsymbol{y}_u)|^{1/\phi} \,du \right) \\
        &= |Q_0| \exp\left[ - \frac{1}{\overline{\kappa}^{1/\phi}}\int_0^t \frac{|z(u,\boldsymbol{y}_u)|^{1/\phi}(T - u +A^{-1/\phi})}{(T - u+A^{-1/\phi})} \,du \right] \\
        &\leqslant |Q_0| \exp\left[ - \frac{\ell^{1/\phi}}{\overline{\kappa}^{1/\phi}}\int_0^t \frac{1}{(T - u +A^{-1/\phi})} \,du \right] \\
        &\leqslant |Q_0| \exp\left[ - \frac{\ell^{1/\phi}}{\overline{\kappa}^{1/\phi}} \log\left( \frac{T+A^{-1/\phi} }{T-t+A^{-1/\phi}} \right) \right],
    \end{split}
\end{align}
from where the result we stated immediately follows.
\end{proof}

Since 
$$
\lim_{A\rightarrow \infty} \left( \frac{A^{-1/\phi} }{T + A^{-1/\phi}} \right)^{\left( \frac{\ell}{\overline{\kappa}}\right)^{1/\phi}} = 0,
$$
we deduce from \eqref{eq:InvBound} that, for any $0 \leqslant \theta < 1,$ we can choose $A$ sufficiently large so as to have $\left|Q_T^*\right| < (1-\theta) |q|.$ More precisely, as long as 
$$
\left( \frac{A^{-1/\phi} }{T + A^{-1/\phi}} \right)^{\left( \frac{\ell}{\overline{\kappa}}\right)^{1/\phi}} < 1 -\theta,
$$
or equivalently,
$$
A > \left\{ \frac{1}{T}\exp\left[ \left( \frac{\ell}{ \overline{\kappa}} \right)^{-1/\phi} \log\left( \frac{1}{1-\theta} \right) \right] \right\}^{\phi}
$$
we guarantee the execution of the fraction $1-\theta$ of the initial inventory. In practice, we can choose $\theta$ in such a way that $(1-\theta)|q|$ is less than one lot size of the asset, resulting in a full execution. We also notice that the bound \eqref{eq:InvBound} is independent of the particular dynamics we assume for the price of the asset, as we can derive it only by assuming that the trader follows strategy $\nu^*$ of \eqref{eq:OptimalSpeed}.

From Theorem \ref{thm:OptimalProcs}, it promptly follows that our optimal strategy does not lead the agent to engage in speculative trading. This is the content of the next result.
\begin{corollary}
The optimal strategy $\left\{ \nu^*_t \right\}_{0\leqslant t \leqslant T}$ does not practice price manipulation, i.e., for $0 \leqslant t \leqslant T,$
\begin{equation} \label{eq:NoPriceManipulation}
    Q_0\,\nu^*_t \leqslant 0\text{ and } Q_0\,Q^*_t \geqslant 0, \, \mathbb{P}-\text{almost surely.}
\end{equation}
\end{corollary}

The first inequality in \eqref{eq:NoPriceManipulation} means that the trader does not buy (respectively, sell) in the context of a liquidation (respectively, acquisition) program. The second one describes that such an agent will not oversell (respectively, overbuy) when executing a portfolio liquidation (respectively, acquisition). Hence, this result does indeed guarantee the absence of price manipulation in the current model.

\subsection{Some numerical experiments}

Our proof of the existence of the solution $z$ of (\ref{eq:ActualMainPDE}) also establishes the convergence of the following numerical algorithm:

\begin{algorithm}[H] \label{algo:NumericalAlgo}
\SetAlgoLined
\KwResult{Numerical solution of (\ref{eq:ActualMainPDE}).}
 Initialize $z^{(0)} := \overline{z}$ or $z^{(0)} := \underline{z},$ $k=0,$ the error variable $\epsilon,$ and stipulate the tolerance $\epsilon_0$\;
 \While{$\epsilon \geqslant \epsilon_0$}{
    $1.
    \begin{cases}
    \partial_t z^{(k+1)} + \mathcal{L}z^{(k+1)} + c z^{(k+1)} =  \gamma \sigma^{1+\phi} - \phi \kappa^{-\frac{1}{\phi}}\left|z^{(k)}\right|^{1+\frac{1}{\phi}} + c z^{(k)} &\text{ in } \left]0,T\right[\times \mathbb{R}^d,\\
    z^{(k+1)}|_{t=T} = -A &\text{ in } \mathbb{R}^d;
    \end{cases}
    $\\
    $2.$ Update $\epsilon;$\\
    $3.$ $k \gets k+1.$
 }
 \Return{$z^k$}
 \caption{Iterative numerical algorithm for solving the PDE (\ref{eq:ActualMainPDE}).}
\end{algorithm}

\begin{remark}
We will show in Section \ref{sec:singular} that the singular limit as $A \rightarrow \infty$ of \eqref{eq:ActualMainPDE} yields the solution of the strict execution problem. Therefore, for large enough $A,$ Algorithm \ref{algo:NumericalAlgo} provides us an approximation for this solution, which is a remarkable advance in the direction of the limitation exposed in \cite[Remark 2.9]{horst2020continuous}.
\end{remark}

At a first step, we notice that the initial iterate $z^{(0)}$ must itself be numerically computed, by using a proper ODE integrator, cf. \eqref{eq:SubSolution} and \eqref{eq:SuperSolution}. There is an easier case, namely when $\phi =1,$ corresponding to a linear temporary price impact setting. In this case, it is straightforward to derive closed-form formulas for both $\underline{z}$ and $\overline{z}.$ Furthermore, each iteration we make in Algorithm \ref{algo:NumericalAlgo}, involves solving a linear parabolic PDE. 

In the numerical experiments that follow, we used a Crank-Nicolson scheme to solve the linear PDE at each iteration step, determining the boundary conditions in the computation domain by linear extrapolation. Here, we make the simplifying assumption of coordinated variation, see \cite[Subsection 1.3]{almgren2012optimal} -- hence, we have $d=m=1.$


We show in Table \ref{tab:Parameters} the parameters that we kept fixed in the numerical experiments that follow. We will describe the remaining ones in each of the corresponding plots. Also, the spatial domain we chose to compute the solution in each of the experiments is $\left[\underline{\boldsymbol{y}},\overline{\boldsymbol{y}}\right] = \left[-5,5\right].$

\begin{table}[!htp]
\centering
\begin{tabular}{@{}cccccc@{}}
\toprule
$T$ & $\boldsymbol{\alpha}\left( \boldsymbol{y} \right)$ & $\boldsymbol{\beta}\left( \boldsymbol{y} \right)$ & $\kappa\left(\boldsymbol{y}\right)$                                                                          & $\sigma\left(\boldsymbol{y}\right)$                                              & $Q^*_0 = q$ \\ \midrule
$5$ & $-5 \boldsymbol{y}$                   & $1$                                  & $\underline{\kappa} \vee \left[ \left( \kappa_0 e^{\boldsymbol{y}} \right)\wedge \overline{\kappa} \right] $ & $\left( \frac{ \kappa_0 }{ \kappa\left( \boldsymbol{y} \right) } \right)^{-1/2}$ & $15$        \\ \bottomrule
\end{tabular}
\caption{Some fixed model parameters we use throughout all of the present simulations. Above, we fix $\kappa_0 := 0.5$ and the mild caps $\underline{\kappa} := \kappa_0 / 10,$ and $\overline{\kappa} := \kappa_0\times 10^4.$}
\label{tab:Parameters}
\end{table}

In Figure \ref{fig:StockVolTemp}, we showcase the particular realization of the stock price corresponding to a volatility and a temporary impact parameter paths that we will use to illustrate the behavior of the strategies. We carry out some comparative statics, varying $A$ in Figure \ref{fig:A}, $\phi$ in Figure \ref{fig:phi}, and $\gamma$ in Figure \ref{fig:gamma}, \textit{ceteris paribus}. We carry out a Monte Carlo simulation of $10^4$ such paths, and demonstrate in Figures \ref{fig:HistogramForGamma} and \ref{fig:HistogramsForPhi} some histograms to illustrate the behavior of the optimal strategies corresponding to each of the values of $\phi$ and $\gamma$ we considered previously. Of course, the same innovations were used for the different parameter values. For $A,$ we find more insightful to understand how the terminal inventory $Q_T^*$ changes with this parameter, whence we plot in Figure \ref{fig:HistogramForA} the histogram of the values of $Q^*_T,$ resulting from this same $10^4$ simulations, for $A\in \left\{3,10\right\}.$  

\begin{figure}[!htp]
    \centering
    \includegraphics[scale = .4]{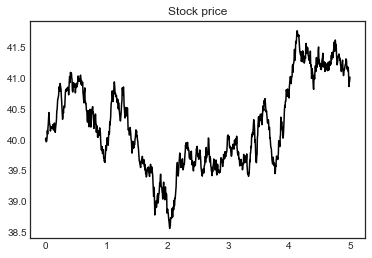}
    \includegraphics[scale = .4]{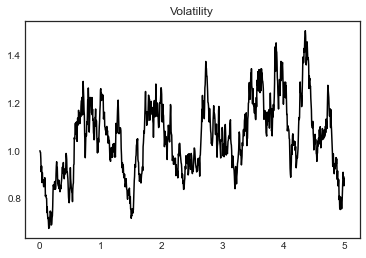}
    \includegraphics[scale = .4]{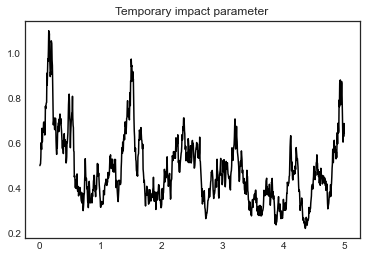}
    \caption{The paths of the stock price, volatility, and temporary impact parameter we used to illustrate the behavior of the strategies in the comparative statics}
    \label{fig:StockVolTemp}
\end{figure}

\begin{figure}[!htp]
    \centering
    \includegraphics[scale = .4]{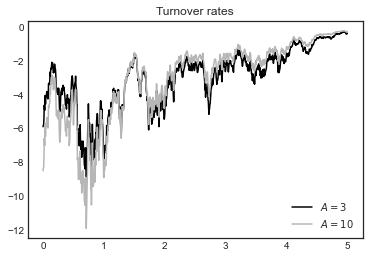}
    \includegraphics[scale = .4]{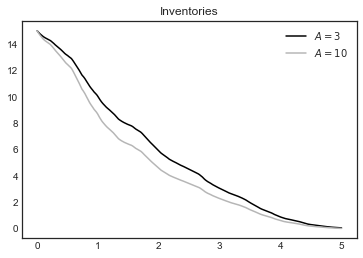}
    \includegraphics[scale = .4]{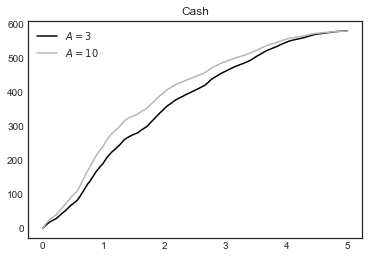}
    \caption{Variation of the strategy and corresponding cash and inventory processes with respect to $A.$ We fixed $\phi = 0.75$ and $\gamma = 0.05.$} 
    \label{fig:A}
\end{figure}

\begin{figure}[!htp]
    \centering
    \includegraphics[scale = .4]{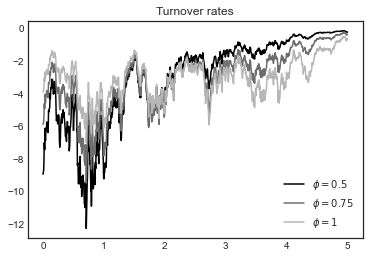}
    \includegraphics[scale = .4]{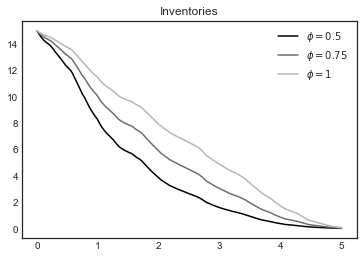}
    \includegraphics[scale = .4]{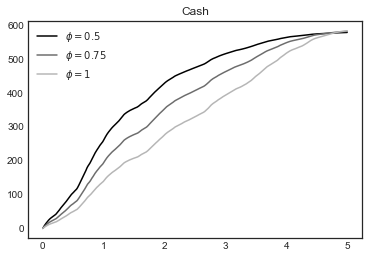}
    \caption{Variation of the strategy and corresponding cash and inventory processes with respect to $\phi.$ We fixed $A = 3$ and $\gamma = 0.05.$} 
    \label{fig:phi}
\end{figure}

\begin{figure}[!htp]
    \centering
    \includegraphics[scale = .4]{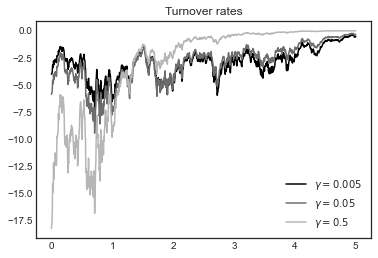}
    \includegraphics[scale = .4]{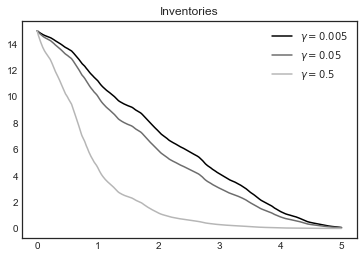}
    \includegraphics[scale = .4]{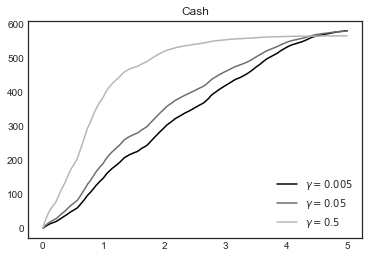}
    \caption{Variation of the strategy and corresponding cash and inventory processes with respect to $\gamma.$ We fixed $A = 3$ and $\phi = 0.75.$} 
    \label{fig:gamma}
\end{figure}

\begin{figure}[!htp]
    \centering
    \includegraphics[scale = .4]{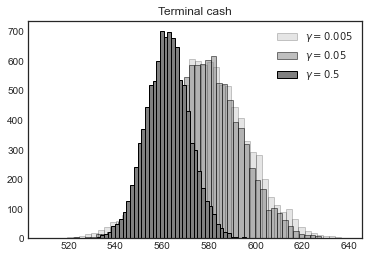}
    \includegraphics[scale = .4]{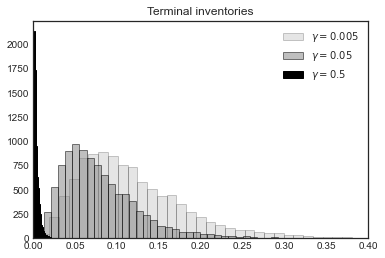}
    \includegraphics[scale = .4]{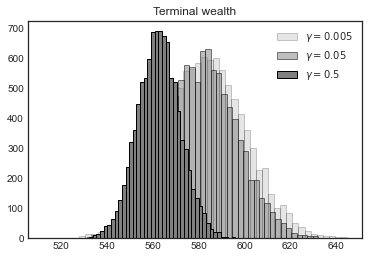}
    \caption{Histograms of $X_T^*,\, Q_T^*$ and $w_T^* = X^*_T + Q^*_T S_T$ for three different values of $\gamma,$ \textit{ceteris paribus}, resulting from the $10^4$ simulations. We fixed $A=3$ and $\phi = 0.75.$ We present the averages and corresponding standard deviations that we computed from our simulations in Table \ref{tab:AvgsAndStdsGammas}.}
    \label{fig:HistogramForGamma}
\end{figure}

\begin{table}[!htp]
\centering
\begin{tabular}{@{}cccc@{}}
\toprule
$\gamma$   & $\mathbb{E}\left[ X^*_T \right]$ & $\mathbb{E}\left[ Q^*_T \right]$ & $\mathbb{E}\left[ w^*_T \right]$ \\ \midrule
$ 5\times 10^{-3}$ & $ 578.446\,(17.587 )$               & $ 0.118\,(0.066 )$                  & $583.184 \,(17.542 )$               \\
$5\times 10^{-2}$  & $578.122 \,(15.612 )$               & $ 0.080\,(0.047 )$                  & $ 581.335\,(15.601 )$               \\
$5\times 10^{-1}$        & $ 562.211\,(9.368 )$                 & $ 0.004\,(0.004 )$               & $ 562.366\,(9.370 )$                \\ \bottomrule
\end{tabular}
\caption{Average values of $X^*_T,\, Q^*_T$ and $w_T^* = X^*_T + Q^*_T S_T,$ computed over all $10^4$ paths we simulated, for some values of $\gamma.$ Here, we fixed $A=3$ and $\phi = 0.75.$ We have put the corresponding standard deviations within parentheses. }
\label{tab:AvgsAndStdsGammas}
\end{table}

\begin{figure}[!htp]
    \centering
    \includegraphics[scale = .4]{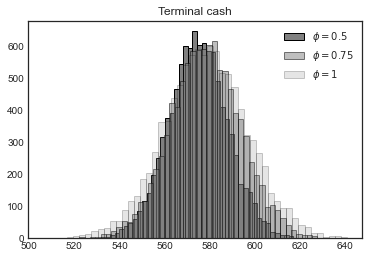}
    \includegraphics[scale = .4]{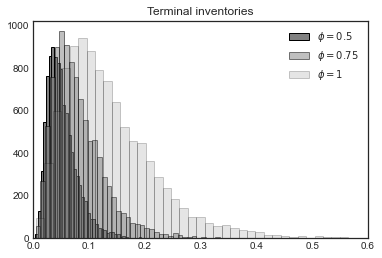}
    \includegraphics[scale = .4]{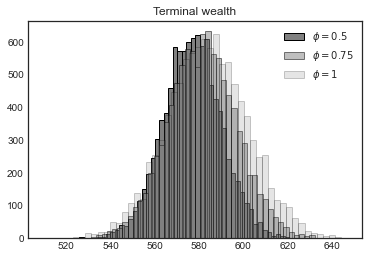}
    \caption{Histograms of $X_T^*,\, Q_T^*$ and $w_T^*$ for three different values of $\phi,$ \textit{ceteris paribus}, resulting from the $10^4$ simulations. We fixed $A=3$ and $\gamma = 0.05.$ We notice that, if we increase $\phi$ and maintain all else equal, the behavior of the trader is to slow down (cf. Figure \ref{fig:phi}). Therefore, she actually takes less impact over the trading schedule, accumulating slightly higher revenues, but has the downside of reaching terminal time holding a larger inventory position. We present the average values resulting from these simulations, as well as their corresponding standard deviations, in Table \ref{tab:AvgsAndStdsPhis}.}
    \label{fig:HistogramsForPhi}
\end{figure}

\begin{table}[!htp]
\centering
\begin{tabular}{@{}cccc@{}}
\toprule
$\phi$ & $\mathbb{E}\left[ X^*_T \right]$ & $\mathbb{E}\left[ Q^*_T \right]$ & $\mathbb{E}\left[ w^*_T \right]$ \\ \midrule
$5 \times 10^{-1}$  & $ 574.683\,( 13.083)$               & $ 0.052\,(0.028 )$                  & $ 576.751\,( 13.097)$                \\
$7.5 \times 10^{-1}$ & $ 578.122\,( 15.612)$               & $ 0.080\,( 0.047)$                   & $ 581.335\,( 15.601)$                \\
$1$    & $ 578.215\,( 18.377)$               & $ 0.137\,( 0.084)$                   & $ 583.694\,( 18.238)$                \\ \bottomrule
\end{tabular}
\caption{Average values of $X^*_T,\, Q^*_T$ and $w_T^*,$ computed over all $10^4$ paths we simulated, for some values of $\phi.$ We fixed $A=3$ and $\gamma = 0.05.$ We have put the corresponding standard deviations within parentheses. }
\label{tab:AvgsAndStdsPhis}
\end{table}

\begin{figure}[!htp]
    \centering
    \includegraphics[scale = .5]{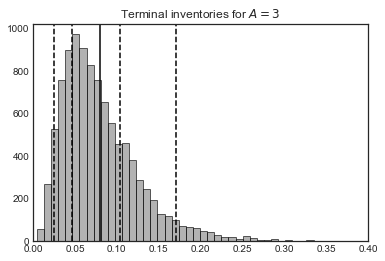}
    \includegraphics[scale = .5]{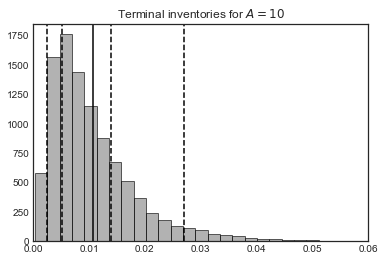}
    \caption{Histograms of the terminal optimal inventory holdings $Q^*_T$ for the values of $A\in \left\{3,10\right\}.$ We fixed $\phi = 0.75$ and $\gamma  = 0.05.$ The dashed lines represent, from the left to the right within each panel, the $5\%,\, 25\%,\, 75\%$ and $95\%$ quantiles. Each solid vertical line lies on the respective average value, taken over all paths we simulated.} 
    \label{fig:HistogramForA}
\end{figure}

\section{Analysis of the constrained problem} \label{sec:singular}

\subsection{The singular limit relative to the terminal penalization parameter}

From now on, corresponding to each $A$ satisfying (\textbf{H3}), let us denote the solution of \eqref{eq:ActualMainPDE} by $z^A.$ We also write $\left\{\nu^{*A}_t\right\}_{0\leqslant t \leqslant T} \in \mathcal{U}_0$ and $\left\{ Q^{*A}_t \right\}_{0 \leqslant t \leqslant T}$ to represent the optimal strategy and inventory holdings, respectively, corresponding to $z=z^A.$ Finally, we set $J_A^\nu(q,\boldsymbol{y}) := z^A(0,\boldsymbol{y})|q|^{1+\phi},$ i.e., $J_A^\nu(q,\boldsymbol{y})$ is the objective criteria \eqref{eq:PerfCrit} associated to the parameter $A,$ $t=0,$ and the strategy $\nu \in \mathcal{U}_0.$ We observe that
$$
J_A^\nu(q,\boldsymbol{y}) = J^\nu(q,\boldsymbol{y}) \hspace{1.0cm} \left(\nu \in \mathcal{U}_c\right).
$$
We refer to Subsection \ref{subsec:PerfCritSing} for the definitions of the performance criteria $J_{\infty}^\nu,$ as well as the set of admissible controls for the constrained problem $\mathcal{U}_c.$ We define the value function
$$
J_{\infty} := \sup_{\nu \in \mathcal{U}_c} J_{\infty}^\nu.
$$
In the subsequent result, we will use another monotonicity argument to derive asymptotic properties of the solution of the PDE \eqref{eq:ActualMainPDE}, as $A \rightarrow \infty.$

\begin{lemma} \label{lem:MonotonicityInA}
Given $(t,\boldsymbol{y}) \in \left[0,T\right]\times \mathbb{R}^d,$ the mapping $A \mapsto z^A(t,\boldsymbol{y})$ is strictly decreasing.
\end{lemma}
\begin{proof}
Let us consider $ A < A^\prime $ (both constrained to (\textbf{H3})). We set $\delta^A := (A-A^\prime)(z^A - z^{A^\prime}).$ We introduce the function
$$
g^{A,A^\prime}(t,\boldsymbol{y}):=
\begin{cases}
\kappa(\boldsymbol{y}) \left( \frac{H\left(z^A(t,\boldsymbol{y})/\kappa( \boldsymbol{y})\right) - H\left(z^{A^\prime}(t,\boldsymbol{y})/\kappa( \boldsymbol{y})\right)}{z^A(t,\boldsymbol{y}) - z^{A^\prime}(t,\boldsymbol{y})(t,\boldsymbol{y})} \right) &\text{ if } z^A(t,\boldsymbol{y}) \neq z^{A^\prime}(t,\boldsymbol{y}),\\
H^\prime\left( z^A(t,\boldsymbol{y})/\kappa( \boldsymbol{y}) \right) &\text{ otherwise }.
\end{cases}
$$
We observe that the function $g^{A,A^\prime}$ is continuous and bounded; hence, it is straightforward to check that $\delta^A$ the unique bounded and continuous viscosity solution of the PDE
$$
\begin{cases}
    \partial_t \delta^A + \mathcal{L}\delta^A +g^{A,A^\prime} \delta^A = 0 &\text{ in } \left]0,T\right[\times \mathbb{R}^d,\\
    \delta^A|_{t=T} = -(A-A^\prime)^2 &\text{ in } \mathbb{R}^d.
\end{cases}
$$
We fix $(t,\boldsymbol{y}) \in \left[0,T\right] \times \mathbb{R}^d$ arbitrarily. We can apply Theorem \ref{thm:Viscosity} to represent $\delta^A$ in the form
$$
\delta^A(t,\boldsymbol{y}) = \mathbb{E}_{t,\boldsymbol{y}}\left[ - (A - A^\prime)^2 e^{\int_t^T g^{A,A^\prime}(\tau,\boldsymbol{y}_\tau)\,d\tau} \right] < 0,
$$
which is clearly equivalent to
$$
z^A(t,\boldsymbol{y}) > z^{A^\prime}(t,\boldsymbol{y}).
$$
This proves the Lemma.
\end{proof}

As a consequence of Lemma \ref{lem:MonotonicityInA}, we can prove that the limit $z^\infty$ of the sequence of functions $\left\{ z^A \right\}_{A}$ is a viscosity solution of the singular problem.
\begin{corollary} \label{cor:DefnOfHinfty}
The function $z^\infty : \left[0,T\right[\times \mathbb{R}^d \rightarrow \mathbb{R}$ defined as
\begin{equation} \label{eq:Limitof_z_inA}
    z^\infty(t,\boldsymbol{y}) := \lim_{A \rightarrow \infty}z^A(t,\boldsymbol{y}) \hspace{1.0cm} \left( (t,\boldsymbol{y}) \in \left[0,T\right[\times \mathbb{R}^d \right)
\end{equation} 
is subject to
\begin{equation} \label{eq:TerminalCondn_zInfty}
    \frac{1}{C(T-t)^\phi}\leqslant \left| z^\infty(t,\,\boldsymbol{y)} \right| \leqslant \frac{C}{(T-t)^\phi} \hspace{1.0cm} \left( (t,\boldsymbol{y}) \in \left[0,T\right[\times \mathbb{R}^d \right),
\end{equation}
and it is a viscosity solution of
\begin{equation} \label{eq:PDE_zInfty}
    \begin{cases}
        \partial_t z^\infty + \mathcal{L}z^\infty + \kappa H\left(z^\infty/\kappa\right) -  \gamma \sigma^{1+\phi} = 0 &\text{ in } \left[0,T\right[\times \mathbb{R}^d,\\
        z^\infty(t,\boldsymbol{y}) \rightarrow - \infty, &\text{ as } t \uparrow T.
    \end{cases}
\end{equation}
Moreover, if $z^\infty$ is continuous, then the limit \eqref{eq:Limitof_z_inA} is uniform over compact subsets of $\left[0,T\right[\times \mathbb{R}^d.$
\end{corollary}
\begin{proof}
The function $z^\infty$ is well-defined by Lemma \ref{lem:MonotonicityInA}, and the relations \eqref{eq:TerminalCondn_zInfty} follow from \eqref{eq:SubAndSuperBlowupRate}. We can show that it is a (possibly discontinuous) viscosity solution of \eqref{eq:PDE_zInfty} using standard stability arguments, such as in \cite[Theorem 6.8]{touzi2012optimal}. If $z^\infty$ is continuous, then Dini's Theorem implies that the convergence \eqref{eq:Limitof_z_inA} is in fact uniform over compact subsets of $\left[0,T\right[\times \mathbb{R}^d.$
\end{proof}

We write $\left\{ \nu^\infty_t\right\}_{0\leqslant t \leqslant T}$ to denote the control given in feedback form by
$$
\nu^\infty(t,q,\boldsymbol{y}) = -\left( - \frac{z^\infty(t,\boldsymbol{y})}{ \kappa(\boldsymbol{y})} \right)^{\frac{1}{\phi}}q,
$$
and by $\left\{ Q^\infty_t \right\}_{0\leqslant t \leqslant T}$ its corresponding inventory process. 

\begin{remark}\label{rem:InvAndRateInfinity}
Proceeding as in Theorem \ref{thm:OptimalProcs}, we derive
$$
Q^\infty_t = q \exp\left( - \int_0^t \left( - \frac{z^\infty(u,\boldsymbol{y}_u)}{ \kappa(\boldsymbol{y}_u)} \right)^{\frac{1}{\phi}} \, du \right),
$$
as well as 
$$
\nu^\infty_t = - q \left( - \frac{z^\infty(t,\boldsymbol{y}_t)}{ \kappa(\boldsymbol{y}_t)} \right)^{\frac{1}{\phi}} \exp\left( - \int_0^t \left( - \frac{z^\infty(u,\boldsymbol{y}_u)}{ \kappa(\boldsymbol{y}_u)} \right)^{\frac{1}{\phi}} \, du \right).
$$
We emphasize that, from the definition of $z^\infty$ in Corollary \ref{cor:DefnOfHinfty}, together with the fact that $z^A \leqslant 0,$ we infer that $z^\infty \leqslant 0.$ Therefore, the limiting strategy $\nu^\infty$ does not lead to price manipulation: $q\,\nu^\infty_t \leqslant 0$ and $q\,Q^\infty_t \geqslant 0,$ for $0 \leqslant t \leqslant T.$
\end{remark}

\subsection{The solution of the constrained problem}

Next, we proceed to prove some convergence results.
\begin{lemma} \label{lem:ConvergencesInvAndRateA}
The following limits hold $\mathbb{P}-$a.s., as $A \rightarrow \infty:$
$$
\begin{cases}
qQ^{*A}_t \downarrow qQ^\infty_t, &\text{ for } 0\leqslant t \leqslant T; \\
\nu^{*A}_t \rightarrow \nu^\infty_t, &\text{ for } 0\leqslant t \leqslant T.
\end{cases}
$$
\end{lemma}
\begin{proof}
By Theorem \ref{thm:OptimalProcs}, the Monotone Convergence Theorem and Remark \ref{rem:InvAndRateInfinity}, we have
$$
qQ^{*A}_t = q^2 \exp\left( - \int_0^t \left( - \frac{z^A(u,\boldsymbol{y}_u)}{ \kappa(\boldsymbol{y}_u)} \right)^{\frac{1}{\phi}} \, du \right) \downarrow q^2 \exp\left( - \int_0^t \left( - \frac{z^\infty(u,\boldsymbol{y}_u)}{ \kappa(\boldsymbol{y}_u)} \right)^{\frac{1}{\phi}} \, du \right) = qQ^\infty_t,
$$
for $0 \leqslant t \leqslant T.$ Similarly,
$$
\nu^{*A}_t = -\left( - \frac{z^A(t,\boldsymbol{y}_t)}{ \kappa(\boldsymbol{y}_t)} \right)^{\frac{1}{\phi}} Q^{*A}_t \xrightarrow{A \rightarrow \infty} -\left( - \frac{z^\infty(t,\boldsymbol{y}_t)}{ \kappa(\boldsymbol{y}_t)} \right)^{\frac{1}{\phi}} Q^{\infty}_t.
$$
\end{proof}

We finish this section by proving that $\left\{ \nu^\infty_t \right\}_{0\leqslant t \leqslant T}$ is the optimal trading rate for the constrained control problem.
\begin{theorem}
The process $\left\{ \nu^\infty_t \right\}_{0\leqslant t \leqslant T}$ belongs to $\mathcal{U}_c,$ and it is the optimal control for the constrained problem.
\end{theorem}
\begin{proof}
We observe that $\mathcal{U}_c \neq \emptyset.$ In effect, an element in it is the TWAP strategy $\nu^{\text{TWAP}}:$
$$
\nu^{\text{TWAP}}_t := \frac{q}{T} \hspace{1.0cm} (0\leqslant t \leqslant T).
$$
We have
$$
J_\infty^{\nu^{\text{TWAP}}}(q,\boldsymbol{y}) = \mathbb{E}_{0,\boldsymbol{y}}\left[\int_0^T \left\{-\kappa(\boldsymbol{y}_t)\left( \frac{|q|}{T} \right)^{1+\phi} -  \gamma \sigma^{1+\phi}(\boldsymbol{y}_t)|q|^{1+\phi}\left(1 - \frac{t}{T}\right)^{1+\phi} \right\} \right] =: m > -\infty.
$$
We emphasize that $m$ is independent of $A.$ Given $\left\{ \nu_t \right\}_{0\leqslant t \leqslant T} \in \mathcal{U}_c \subseteq \mathcal{U}_0$ arbitrarily, it follows that
\begin{equation} \label{eq:SingProblem1}
    J^{\nu^{*A}}_A( q,\boldsymbol{y}) \geqslant J_A^{\nu}(q,\boldsymbol{y}) = J_\infty^{\nu}(q,\boldsymbol{y}).
\end{equation}
In particular,
$$
J_A^{\nu^{*A}}(q,\boldsymbol{y}) \geqslant m,
$$
whence
\begin{align*}
    -\limsup_{A\rightarrow \infty} \mathbb{E}\left[ \int_0^T \left| \nu^{*A}_t \right|^{1+\phi} \,dt \right] &= \liminf_{A\rightarrow \infty} \mathbb{E}\left[ -\int_0^T \left| \nu^{*A}_t \right|^{1+\phi} \,dt \right] \\
    &\geqslant \frac{1}{\underline{\kappa}}\liminf_{A\rightarrow \infty} \mathbb{E}\left[ -\int_0^T \kappa(\boldsymbol{y}_t)\left| \nu^{*A}_t \right|^{1+\phi} \,dt \right] \\
    &\geqslant \frac{1}{\underline{\kappa}}\liminf_{A \rightarrow \infty} J_A^{\nu^{*A}}(0,q,\boldsymbol{y}) \\
    &\geqslant \frac{m}{\underline{\kappa}}.
\end{align*}

Therefore, we can employ Lemma \ref{lem:ConvergencesInvAndRateA} and Fatou's Lemma to infer that
$$
\mathbb{E}\left[\int_0^T \left| \nu^\infty_t \right|^{1+\phi}\,dt \right]\leqslant \liminf_{A\rightarrow \infty} \mathbb{E}\left[ \int_0^T \left| \nu^{*A}_t \right|^{1+\phi} \,dt \right] \leqslant - \frac{m}{\underline{\kappa}} < \infty.
$$
This proves that $\left\{\nu^\infty_t \right\}_{0\leqslant t \leqslant T} \in \mathbb{L}^{1+\phi}.$ 

We now turn to the proof of the fact that $\left\{ \nu^\infty_t \right\}_t \in \mathcal{U}_c.$ To do this, we notice that
$$
A \mathbb{E}_{0,q,\boldsymbol{y}}\left[ \left| Q^A_T \right|^{1+\phi} \right] \leqslant - J^{\nu^{*A}}(0,q,\boldsymbol{y}) \leqslant -m, 
$$
whence, arguing as above, we obtain 
$$
\mathbb{E}_{0,q,\boldsymbol{y}}\left[ \left| Q^\infty_T \right|^{1+\phi} \right] \leqslant \liminf_{A \rightarrow \infty} \mathbb{E}_{0,q,\boldsymbol{y}}\left[ \left| Q^A_T \right|^{1+\phi} \right] = 0.
$$
In this way, we deduce that $Q^\infty_T = 0$ $\mathbb{P}-$a.s., from where it follows that $\left\{ \nu^\infty_t \right\}_t \in \mathcal{U}_c.$

Likewise, for each $\left\{ \nu_t\right\}_{0\leqslant t \leqslant T} \in \mathcal{U}_c,$ we once more employ Lemma \ref{lem:ConvergencesInvAndRateA} and Fatou's Lemma, now together with relation \eqref{eq:SingProblem1}, to infer the following
\begin{align} \label{eq:SingProblem2}
     \begin{split}
         J^{\nu^\infty}_\infty(q,\boldsymbol{y}) &\geqslant \limsup_{A\rightarrow \infty}\mathbb{E}_{0,q,\boldsymbol{y}}\left[-\int_0^T \left\{ \kappa(\boldsymbol{y}_t)\left| \nu^{*A}_t \right|^{1+\phi} + \gamma \sigma^{1+\phi}(\boldsymbol{y}_t)\left|Q^{*A}_t\right|^{1+\phi} \right\}\,dt \right] \\
         &\geqslant \limsup_{A\rightarrow \infty} J^{\nu^{*A}}_A(q,\boldsymbol{y})\\
         &\geqslant J^{\nu}_\infty(q,\boldsymbol{y}).
     \end{split}
\end{align}
Therefore, \eqref{eq:SingProblem2} implies
\begin{equation} \label{eq:ArgmaxSingProblem}
    \left\{ \nu^\infty_t \right\}_{t} \in \argmax_{\nu \in \mathcal{U}_c} J_\infty^{\nu}(q,\boldsymbol{y}).
\end{equation}
In fact, $\left\{ \nu^\infty_t \right\}_{t}$ is the unique solution of \eqref{eq:ArgmaxSingProblem}, since the functional $\nu \in \mathcal{U}_c \mapsto J_\infty^\nu(q,\boldsymbol{y}) \in \mathbb{R}$ is strictly concave.
\end{proof}

\section{Conclusions} \label{sec:Conclusions}

We investigated the problem of optimal portfolio execution under a framework suited to illiquid markets. The market friction we considered took the form of a temporary price impact, determined by the trader's turnover rate according to a power law. Furthermore, we modeled the slope corresponding to this cost as a stochastic process. Likewise, we considered the volatility of the price of the asset to be uncertain. The dynamic assumption we made over these two processes is that their driver is a multidimensional Markov diffusion.

To obtain our optimal trading strategy in the regularized setting, in which we did not require complete execution of the initial inventory, we proposed performance criteria under the Implementation Shortfall paradigm, leading us to derive the HJB PDE that the value function should solve. Under an adequate ansatz, we simplified this PDE. We were able to apply an iterative monotonicity technique to show, under mild model assumptions, that this equation admitted a unique continuous and bounded solution. We proved that the optimal trading rate thus obtained did not lead the agent to engage in speculative trading. Furthermore, our method yielded an iterative algorithm for solving the PDE numerically. We presented a number of numerical experiments.

In the last part of the work, we considered the constrained problem, where we required complete execution of the initial portfolio. We were able to show that the functions determining the optimal strategies in the regularized framework satisfied a monotonicity relation, allowing us to define their pointwise limit. We could retain a certain degree of integrability, when passing the corresponding strategies to the limit, thanks to the form of the performance criteria. Then, we used a comparison argument between the optimal rates of the original framework and the admissible strategies of the constrained one to establish that the limiting strategy is indeed a solution of the latter. The fact that it is the unique one followed from a concavity argument.

\section*{Acknowledgement}

This study was financed in part by Coordena\c{c}\~ao de Aperfei\c{c}oamento de Pessoal de N\'ivel Superior - Brasil (CAPES) - Finance code 001. MOS was partially supported by CNPq grant \# 310293/2018-9.


\bibliographystyle{apalike}
\bibliography{References}


\appendix

\section{On terminal time asymptotics of the sub- and supersolutions} \label{app:BlowupRateOfSubAndSuper}

Under the notations of Subsection \ref{subsec:SubSuper}, we observe that
\begin{equation} \label{eq:BlowupInitialEstimates}
    \frac{1}{b}\int_{-A}^\xi \frac{du}{|u|^r} \leqslant - \int_{-A}^\xi \frac{du}{a - b|u|^r} \leqslant \left(b - \frac{a}{|\xi|^r} \right)^{-1} \int_{-A}^\xi \frac{du}{|u|^r},
\end{equation}
for $-A \leqslant \xi \leqslant -(a/b)^{1/r}.$ Since
$$
\int_{-A}^\xi \frac{du}{|u|^r} = \frac{|\xi|^{1-r} - A^{1-r}}{r-1},
$$
we obtain from the first inequality in \eqref{eq:BlowupInitialEstimates} that $y(t) = F^{-1}(T-t)$ satisfies
\begin{equation} \label{eq:ConcludingGenIneqRates_pt1}
     |y| \geqslant \left[A^{1-r} + b(r-1)(T-t) \right]^{\frac{1}{1-r}}.
\end{equation}
From the second one, we estimate
\begin{align} \label{eq:IntermedStepConcIneq_pt2}
    \begin{split}
        T-t &\leqslant \left(b - \frac{a}{|y|^r} \right)^{-1} \int_{-A}^y \frac{du}{|u|^r} \\
        &= \frac{|y|^r}{b|y|^r - a}\left( \frac{|y|^{1-r} - A^{1-r}}{r-1} \right).
    \end{split} 
\end{align}
Carrying out some simple manipulations on \eqref{eq:IntermedStepConcIneq_pt2}, we deduce that
$$
|y|^{r-1}\leqslant \left[ \frac{a(T-t)(r-1)}{|y|} + 1 \right]\left[b(r-1)(T-t) + A^{1-r} \right]^{-1} \leqslant \left[ \frac{a(T-t)(r-1)}{(a/b)^{1/r}} + 1 \right]\left[b(r-1)(T-t) + A^{1-r} \right]^{-1}, 
$$
from where we conclude
\begin{equation} \label{eq:ConcludingGenIneqRates_pt2}
|y| \leqslant \left[ \frac{a(r-1)T}{(a/b)^{1/r}} + 1 \right]^{\frac{1}{r-1}}\left[b(r-1)(T-t) + A^{1-r} \right]^{\frac{1}{1-r}}.
\end{equation}
From the definitions of $\underline{z}$ and $\overline{z},$ see Subsection \ref{subsec:SubSuper}, the relations we stated in \eqref{eq:SubAndSuperBlowupRate} follow promptly from \eqref{eq:ConcludingGenIneqRates_pt1} and \eqref{eq:ConcludingGenIneqRates_pt2}.  

\section{Continuity of \texorpdfstring{$\underline{z}^*$}{limiting function}} \label{app:ContinuityIssue}

For each $(t,\boldsymbol{y}) \in \left[0,T\right] \times \mathbb{R}^d$ and $t \leqslant u \leqslant T,$ we write $z_u^{t,\boldsymbol{y}} := \underline{z}^*\left( u,\, \boldsymbol{y}_u^{t,\boldsymbol{y}} \right),$ where $\left\{ \boldsymbol{y}_u^{t,\boldsymbol{y}} \right\}_u$ has dynamics \eqref{eq:MarkovDiffusions} and initial condition $\boldsymbol{y}_t^{t,\boldsymbol{y}} = \boldsymbol{y}.$ Given $t_0 \in \left[0,T\right]$ and $\boldsymbol{y}_1,\,\boldsymbol{y}_2 \in \mathbb{R}^d,$ we can carry out simple estimates to show that the quantity
$$
\delta z_t^{t_0} := \sup_{t\leqslant u \leqslant T} \mathbb{E} \left|z_u^{t_0,\boldsymbol{y}_1} - z_u^{t_0,\boldsymbol{y}_2}\right| \hspace{1.0cm} (t_0 \leqslant t \leqslant T)  
$$
satisfies
\begin{align*}
    \delta z_t^{t_0} \leqslant &\,C \mathbb{E}\left[\int_t^T \left| e^{\int_t^u c(\tau,\boldsymbol{y}_\tau^{t_0,\boldsymbol{y}_1})\,d\tau} \sigma^{1+\phi}\left( \boldsymbol{y}_u^{t_0,\boldsymbol{y}_1} \right) - e^{\int_t^u c(\tau,\boldsymbol{y}_\tau^{t_0,\boldsymbol{y}_2})\,d\tau}\sigma^{1+\phi}\left( \boldsymbol{y}_u^{t_0,\boldsymbol{y}_2} \right) \right| \,du \right] \\
    & + C\mathbb{E}\left[\int_t^T\left| e^{\int_t^u c(\tau,\boldsymbol{y}_\tau^{t_0,\boldsymbol{y}_1})\,d\tau} \kappa^{-1/\phi}\left( \boldsymbol{y}_u^{t_0,\boldsymbol{y}_1} \right) - e^{\int_t^u c(\tau,\boldsymbol{y}_\tau^{t_0,\boldsymbol{y}_2})\,d\tau}\kappa^{-1/\phi}\left( \boldsymbol{y}_u^{t_0,\boldsymbol{y}_2} \right) \right|\,du \right] \\
    & + C \mathbb{E}\left[ \int_t^T \left| c(u,\boldsymbol{y}_u^{t_0,\boldsymbol{y}_1}) - c(u,\boldsymbol{y}_u^{t_0,\boldsymbol{y}_2}) \right|\,du + \left| e^{\int_t^T c(u,\boldsymbol{y}_u^{t_0,\boldsymbol{y}_1})\,d\tau} - e^{\int_t^T c(u,\boldsymbol{y}_u^{t_0,\boldsymbol{y}_2})\,du} \right| \right] \\
    & + C\int_t^T \delta z^{t_0}_u\,du,
\end{align*}
for $t_0 \leqslant t \leqslant T,$ whence Gronwall's Lemma implies
$$
\left| \underline{z}^*(t_0,\boldsymbol{y}_1) - \underline{z}^*(t_0,\boldsymbol{y}_2) \right| \leqslant \delta z_{t_0}^{t_0} \leqslant C \omega(t_0,\boldsymbol{y}_1,\boldsymbol{y}_2),
$$
for a suitable continuous function $\omega : \left[0,T\right] \times \mathbb{R}^d \times \mathbb{R}^d \rightarrow \mathbb{R}$ such that $\omega(t_0,\boldsymbol{y},\boldsymbol{y}) =0.$ We emphasize that we can prove the continuity of $\omega$ through standard arguments using basic SDE estimates and the Dominated Convergence Theorem, cf. the proof of the continuity part of \cite[Theorem 3.42]{pardoux2014stochastic}. Likewise, we show that $t \in \left[0,T\right] \mapsto \underline{z}^*(t,\boldsymbol{y}) \in \mathbb{R}$ is continuous, for each $\boldsymbol{y} \in \mathbb{R}^d$ (in fact, locally uniformly in the spatial variable). In this way, we conclude that 
$$
\left| \underline{z}^*(t_1,\boldsymbol{y}_1) - \underline{z}^*(t_2,\boldsymbol{y}_2) \right| \leqslant  \left| \underline{z}^*(t_1,\boldsymbol{y}_1) - \underline{z}^*(t_2,\boldsymbol{y}_1) \right| + \left| \underline{z}^*(t_2,\boldsymbol{y}_1) - \underline{z}^*(t_2,\boldsymbol{y}_2) \right| \rightarrow 0,
$$
as $(t_2,\boldsymbol{y}_2) \rightarrow (t_1,\boldsymbol{y}_1).$

\section{Proof of the verification result} \label{app:Verification}

We notice that the strategy $\left\{ \nu^*_t \right\}_t$ is clearly admissible, as we observe from its definition \eqref{eq:OptimalSpeed} that it is in fact uniformly bounded. Also, $\left|Q^*_t \right| \leqslant |q|.$ Let $\left\{ \left( U^{t,\boldsymbol{y}}_u,\,Z_u^{t,\boldsymbol{y}}\right)\right\}_{t\leqslant u \leqslant T}$ be the solution of the BSDE
$$
dU^{t,\boldsymbol{y}}_u = - \Phi(\boldsymbol{y}^{t,\boldsymbol{y}}_u,z\left( u,\boldsymbol{y}^{t,\boldsymbol{y}}_u\right) )\,du + Z_u^{t,\boldsymbol{y}}\,d\boldsymbol{W}_u,\, U_T^{t,\boldsymbol{y}} = -A.
$$
where
$$
\Phi\left(\boldsymbol{y}, z \right) = \gamma \sigma^{1+\phi}(\boldsymbol{y}) - \kappa H\left(z/\kappa\right),
$$
and
$$
d\boldsymbol{y}^{t,\boldsymbol{y}}_u = \boldsymbol{\alpha}(\boldsymbol{y}^{t,\boldsymbol{y}}_u)\,du + \boldsymbol{\beta}(\boldsymbol{y}^{t,\boldsymbol{y}}_u)d\boldsymbol{W}_u,\, \boldsymbol{y}^{t,\boldsymbol{y}}_t = \boldsymbol{y}.
$$
Explicitly, 
$$
U^{t,\boldsymbol{y}}_u = \mathbb{E}_{u,\, \boldsymbol{y}_u^{t,\boldsymbol{y}}}\left[ \int_u^T \Phi\left( \boldsymbol{y}^{t,\boldsymbol{y}}_r,\, z\left(r,\, \boldsymbol{y}^{t,\boldsymbol{y}}_r\right)  \right)\,dr - A\right] .
$$
By the Feynman-Kac formula, see Theorem \ref{eq:FeynmanKac}, the bounded function $w(t,\boldsymbol{y}) := U^{t,\boldsymbol{y}}_t$ solves, in the viscosity sense, a linear PDE that $z$ also turns out to solve; hence, by Corollary \ref{cor:Comparison}, we infer that $w \equiv z.$ By the Markov property, we have $U^{t,\boldsymbol{y}}_u = z(u,\boldsymbol{y}_u^{t,\boldsymbol{y}}),$ for $t\leqslant u \leqslant T.$ Next, we apply It\^o's formula to $\left\{ U^{t,\boldsymbol{y}}_u\left| Q_u^{\nu} \right|^{1+\phi} \right\}_{u\leqslant t \leqslant T},$ for a given $\nu \in \mathcal{U}_t,$ and take expectations to derive
\begin{align} \label{eq:Verification_main_identity}
    \begin{split}
        z(t,\boldsymbol{y})|q|^{1+\phi} =&\, U^{t,\boldsymbol{y}}_t|q|^{1+\phi} \\
        =& \, \mathbb{E}\left[-A|Q_T^{\nu}|^{1+\phi} - \int_t^T \left\{ \kappa\left( \boldsymbol{y}_u^{t,\boldsymbol{y}} \right)|\nu_u|^{1+\phi} +  \gamma \sigma^{1+\phi}\left(\boldsymbol{y}_u^{t,\boldsymbol{y}} \right)|Q^\nu_u|^{1+\phi} \right\}\,du \right] \\
        &+ \mathbb{E}\left[\int_t^T \left\{ \kappa\left(\boldsymbol{y}_u^{t,\boldsymbol{y}}\right) H\left(U^{t,\boldsymbol{y}}_u/\kappa\left(\boldsymbol{y}_u^{t,\boldsymbol{y}} \right) \right)\left| Q^\nu_u \right|^{1+\phi} - \mathcal{H}\left(Q^\nu_u,\, \boldsymbol{y}_u^{t,\boldsymbol{y}},\, U^{t,\boldsymbol{y}}_u,\, \nu_u \right) \right\}\,du \right],
    \end{split}
\end{align}
where
$$
\mathcal{H}\left(q,\, \boldsymbol{y},\, z,\, \nu \right) := (1+\phi)z\left|q\right|^\phi \sign\left( q \right)\nu - \kappa\left( \boldsymbol{y} \right) \left|\nu \right|^{1+\phi}.
$$

On the one hand, we notice that the maximum of $\nu \mapsto \mathcal{H}\left(q,\, \boldsymbol{y},\, z,\, \nu \right)$ is attained at $\nu = -\left(-z/\kappa\right)^{1/\phi}q,$ and furthermore 
$$
\mathcal{H}\left(q,\, \boldsymbol{y},\, z,\, -\left(-z/\kappa\right)^{1/\phi}q \right) = \kappa(\boldsymbol{y}) H\left(z/\kappa\left(\boldsymbol{y} \right) \right)\left| q \right|^{1+\phi},
$$
whence, from \eqref{eq:Verification_main_identity}, we always have
\begin{equation*} 
    z(t,\boldsymbol{y})|q|^{1+\phi} \geqslant \mathbb{E}\left[-A|Q_T^{\nu}|^{1+\phi} - \int_t^T \left\{ \kappa\left( \boldsymbol{y}_u^{t,\boldsymbol{y}} \right)|\nu_u|^{1+\phi} +  \gamma \sigma^{1+\phi}\left(\boldsymbol{y}_u^{t,\boldsymbol{y}} \right)|Q^\nu_u|^{1+\phi} \right\}\,du \right] = J^\nu\left(t,q,\boldsymbol{y}\right),
\end{equation*}
from where it follows that
\begin{equation} \label{eq:Verification_conclusion_pt1}
    z(t,\boldsymbol{y})|q|^{1+\phi} \geqslant \sup_{\nu \in \mathcal{U}_t } J^\nu(t,q,\boldsymbol{y}) =  J\left(t,q,\boldsymbol{y}\right).
\end{equation}
On the other hand, by setting $\nu = \nu^*$ in \eqref{eq:Verification_main_identity}, we infer
\begin{equation} \label{eq:Verification_conclusion_pt2}
    z(t,\boldsymbol{y})|q|^{1+\phi} = J^{\nu^*}\left(t,q,\boldsymbol{y}\right)
\end{equation}
Putting relations \eqref{eq:Verification_conclusion_pt1} and \eqref{eq:Verification_conclusion_pt2} together, we establish the verification result.


\end{document}